\documentclass[acmsmall]{acmart}

\AtBeginDocument{%
	\providecommand\BibTeX{{%
			\normalfont B\kern-0.5em{\scshape i\kern-0.25em b}\kern-0.8em\TeX}}}

\setcopyright{rightsretained}
\acmDOI{10.1145/3643772}
\acmYear{2024}
\copyrightyear{2024}
\acmSubmissionID{fse24main-p1015-p}
\acmJournal{PACMSE}
\acmVolume{1}
\acmNumber{FSE}
\acmArticle{46}
\acmMonth{7}
\received{2023-09-28}
\received[accepted]{2024-01-23}
 
\usepackage{stmaryrd,amsfonts,dsfont,amsmath}
\usepackage{semantic,multirow,multicol,graphicx,comment}
\usepackage[algoruled,linesnumbered,noend]{algorithm2e}
\usepackage{cleveref}
\usepackage{makecell}
\usepackage{colortbl,booktabs}
\usepackage{listings,framed}
\usepackage{subcaption}
\extrarowheight=\aboverulesep 
\usepackage{enumitem}

\newcommand{\tool}{{\sc CT-Prover}\xspace}

\newcommand{\toollight}{{\sc CT-Prover}$^{-}$\xspace}

\newcommand{\while}{{\sc While}\xspace}

\newcommand{\myblue}[1]{\textcolor{blue}{#1}}
\newcommand{\mycyan}[1]{\textcolor{cyan}{#1}}
\newcommand{\sskip}{\myblue{\tt skip}\xspace}
\newcommand{\sassume}{\myblue{\tt assume}\xspace}
\newcommand{\Prog}{{\tt P}}
\newcommand{\sassert}{\myblue{\tt assert}\xspace}

\newcommand{\sif}{\myblue{\tt if}\xspace}
\newcommand{\sfi}{\myblue{\tt fi}\xspace}
\newcommand{\sthen}{\myblue{\tt then}\xspace}
\newcommand{\selse}{\myblue{\tt else}\xspace}
\newcommand{\swhile}{\myblue{\tt while}\xspace}
\newcommand{\sreturn}{\myblue{\tt return}\xspace}

\newcommand{\bbegin}{{\tt begin}\xspace}
\newcommand{\eexit}{{\tt exit}\xspace}

\newcommand{\sdo}{\myblue{\tt do}\xspace}
\newcommand{\sod}{\myblue{\tt od}\xspace}
\newcommand{\Loc}{\mathbb{L}}
\newcommand{\Val}{\mathbb{V}}

\newcommand{\true}{\myblue{\tt true}\xspace}
\newcommand{\false}{\myblue{\tt false}\xspace}

\newcommand{\sdef}{\myblue{\tt def}\xspace}

\newcommand{\Var}{\mathbb{X}}

\newcommand{\var}{{\tt var}\xspace}
\newcommand{\lfp}{{\tt lfp}\xspace}

\newcommand{\Guard}{\mycyan {\tt Guard}\xspace}

\newcommand{\smain}{{\tt main}}

\newcommand{\obs}{\mathcal{O}}

\newcommand{\Nn}{\mathbb{N}}

\usepackage[many]{tcolorbox}

\begin{document}

\title{Towards Efficient Verification of Constant-Time Cryptographic Implementations}

\author{Luwei Cai}
\orcid{0000-0003-3148-7926}
\affiliation{%
  \institution{ShanghaiTech University}
  \city{Shanghai}
  \country{China}
}
\email{cailw@shanghaitech.edu.cn}

\author{Fu Song}
\orcid{0000-0002-0581-2679}
\authornote{Corresponding author}
\email{songfu@ios.ac.cn}
\affiliation{%
  \institution{State Key Laboratory of Computer Science, Institute of Software, Chinese Academy of Sciences}
  \city{Beijing}
  \country{China}
}
\additionalaffiliation{
\institution{University of Chinese Academy of Sciences, and Nanjing Institute of Software Technology}
}
 
\author{Taolue Chen}
\orcid{0000-0002-5993-1665}
\affiliation{%
  \institution{Birkbeck, University of London}
  \city{London}
  \country{UK}
}
\email{t.chen@bbk.ac.uk}
 
	
\begin{abstract}
Timing side-channel attacks exploit secret-dependent execution time to fully or partially recover secrets of cryptographic implementations, posing a severe threat to software security. Constant-time programming discipline is an effective software-based countermeasure against timing side-channel attacks, but developing constant-time implementations turns out to be challenging and error-prone. Current verification approaches/tools suffer from scalability and precision issues 
when applied to production software in practice. In this paper, we put forward practical verification approaches based on a novel synergy of taint analysis and safety verification of self-composed programs. Specifically, we first use an IFDS-based lightweight taint analysis to prove that a large number of potential (timing) side-channel sources do not actually leak secrets. We then resort to a precise taint analysis and a safety verification approach to 
determine whether the remaining potential side-channel sources can actually leak secrets. These include novel constructions of taint-directed semi-cross-product of the original program and its Boolean abstraction, and a taint-directed self-composition of the program. Our approach is implemented as a cross-platform and fully automated tool {\tool}. The experiments confirm its efficiency and effectiveness in verifying real-world benchmarks from modern cryptographic and SSL/TLS libraries. 
In particular, {\tool} identify new, confirmed vulnerabilities of open-source SSL libraries (e.g., Mbed SSL, BearSSL) and significantly outperforms the state-of-the-art tools.
\end{abstract}

\ccsdesc[500]{Software and its engineering~Formal software verification}
\ccsdesc[500]{Theory of computation~Program analysis}
\ccsdesc[500]{Security and privacy~Formal security models}
\ccsdesc[500]{Security and privacy~Logic and verification}
	
\keywords{Timing side-channel, constant-time cryptographic implementation, formal verification, taint analysis}

 \maketitle
	
\section{Introduction}
\label{sec:intr}
The security of contemporary software systems
and communication heavily depends upon cryptographic implementations, which are the main focus of the current paper. Timing side-channel attacks~\cite{Kocher96,BrumleyT11} can exploit 
secret-dependent execution time to fully or partially recover secrets even remotely, thus posing a severe threat to software security. 
Over the past few years, numerous timing side-channel vulnerabilities have been discovered, allowing adversaries to deduce secrets with very few trials.
Notorious examples include the Lucky 13 attack that can remotely recover plaintext from the CBC-mode encryption in TLS and DTLS (due to an unbalanced branching statement~\cite{AlFardanP13}),
and Brumley and Boneh's remote private key recovery attack against the sliding window exponentiation of the RSA decryption in OpenSSL~\cite{brumley2003remote,BrumleyT11}.
It is vital to implement effective countermeasures to protect cryptographic implementations.

There are different means to mitigate the risk of timing side-channels, for instance, via 
security-aware system (e.g.,~\cite{LiGR14}) and architecture (e.g.,~\cite{DanielBNBRP23}). 
A more effective approach is to 
come up with better software implementation to eliminate the root cause.
Currently, there are two prevailing countermeasures, i.e., \emph{constant-time} and \emph{time-balancing} programming disciplines.
The former requires that control flow and memory-access patterns of an implementation are independent of the secrets; the latter requires the execution time to be negligibly influenced by secrets which can be considered as a tradeoff between security and enforceability.
Both countermeasures have been adopted in open-source cryptographic and SSL/TLS libraries such as NaCl~\cite{BernsteinLS12}, BearSSL~\cite{BearSSL23}, Mbed TLS~\cite{MbedTLS23}, s2n-tls~\cite{s2n23}, and OpenSSL~\cite{openSSL23}.
Writing constant-time or time-balancing implementations requires the use of low-level programming languages or compiler knowledge, and developers usually need to deviate from standard programming practices.  
Furthermore, their correctness depends on global properties across pairs of executions, 
hence is difficult to reason about. For instance, even though two protections against Lucky 13 were implemented in s2n-tls, the Lucky microseconds attack, a variant of Lucky 13, can remotely and completely recover plaintext from the CBC-mode cipher suites in s2n-tls, resulting in a complete recovery of HTTP session cookies and user credentials such as BasicAuth passwords~\cite{AlbrechtP16}. Alas, timing side-channel attacks remain a live threat for cryptographic libraries after their discovery over 25 years ago~\cite{YLW23},
and it is essential to develop automated reasoning tools for formally verifying these countermeasures. 

In this paper, we focus on the constant-time programming discipline as it represents a more fundamental solution and is more challenging to tackle.
Numerous verification approaches have been proposed (cf.\ Section~\ref{sect:related}).  
A majority of them leverage (lightweight) static analysis, e.g., abstract interpretation, type system, taint analysis and (relational) symbolic execution which are sound (and usually efficient) 
but incomplete (i.e., false positives may occur or depth of explored paths is bound). 
Different from them, the self-composition~\cite{ABBDE16} approach 
reduces constant-time verification to safety verification
by building a program consisting of two copies of the given program and verifying whether the values of the low-security (public) variables in the two copies are identical
provided that the public inputs are identical. Self-composition based approaches are sound and 
complete in theory, although in practice, the completeness may depend on verification tools to carry out challenging safety property checking on self-composed programs (cf.~Section~\ref{sect:rq1} for concrete examples).

Self-composition based approaches are normally considered as a heavyweight approach which is not scalable. The primary reason is that they need to deal with a self-composed program of quadratic sizes. 
As a result, many efforts have been made to ease safety verification of self-composed programs, including 
self-composition with lockstep execution of loops~\cite{SousaD16} for $k$-safety properties; 
self-composition with lockstep execution of both loops and branches (called cross-product)~\cite{ABBDE16}
for constant-time properties;
type-directed self-composition~\cite{TerauchiA05} 
and lazy self-composition~\cite{YangVSGM18} for proving information flow properties. The commonality of these approaches is to simplify self-composed programs,
reduce the number of safety checks, and/or keep variables from the two copies near each other. 

Despite these efforts, a noticeable gap exists in verifying constant-time countermeasures in software engineering practice. To the best of our knowledge, {\sf ct-verif}~\cite{ABBDE16} is the only publicly available self-composition based tool that 
is being actively deployed in the continuous integration of Amazon's s2n-tls library~\cite{s2nctverif}, 
but is significantly less efficient than lightweight static analysis approaches (e.g., taint analysis and type system), and often fails to prove constant-time implementations (cf.\ Section~\ref{sec:expe}).  
In summary, the current status is that the user either chooses an incomplete approach but needs to tackle false positives, or chooses a (relatively) complete approach which is costly and may fail on a number of occasions.   

The main purpose of the current paper is to make the self-composition based approaches scalable so they can be used to handle the verification of constant-time cryptographic implementations at the industry level.  
Our strategy towards both completeness and scalability is to take a stratified approach. Technically, starting from potential timing side-channel sources which can be identified straightforwardly (cf.\ Section~\ref{sec:CT}), we devise two different taint analyses and a taint-directed cross-product, and gradually integrate them to resolve these potential sources, i.e., to determine whether they can actually cause information leakage. 

The first taint analysis (cf.~Section~\ref{sec:1sttaint}) leverages the inter-procedural, finite, distributive, subset framework (IFDS~\cite{RHS95}), which is
accelerated by propagating the data-flow facts sparsely~\cite{OhHLLY12}.
This taint analysis is flow-, field- and context-sensitive, but path- and index-insensitive.
It is often able to efficiently prove that a large number of 
potential (timing) side-channel sources do not actually leak secrets, leaving a relatively small number of them unresolved to the next step.

The second taint analysis (cf.~Section~\ref{sect:2ndtaint}) 
uses a novel \emph{taint-directed semi-cross-product}, which reduces the flow-, context-, path-, field- and index-sensitive taint analysis problem to checking safety properties
of the cross-product of the given program and its Boolean abstraction. 
The Boolean abstraction is used to track the required information flow from the secrets.
This precise taint analysis would be able to further resolve many remaining potential side-channel sources.
It is worth noting that our cross-product is taint-directed, namely, it is based on the taint information from the first taint analysis, which greatly reduces the number of safety checks and simplifies the product program, hence improving the verification efficiency. 

Finally, to resolve the remaining (usually few) potential side-channel sources, we propose a \emph{taint-directed cross-product} (cf.~Section~\ref{sec:cpverif}), which reduces the constant-time security problem to the safety problem of the cross-product of the program
where taint information is also used to reduce the number of safety checks and simplify the cross-product program.

We implement our approaches as a fully automated, cross-platform tool \tool for verifying (optimized) LLVM IR implementations. To evaluate \tool, 
we collect 87 real-world implementations from modern cryptographic and SSL/TLS libraries, as well as fixed-point arithmetic libraries. These benchmarks include cryptographic utilities, arithmetic operations, public and private key cryptography,
and algorithms for encryption, decryption, message authentication code, and digital signature.
The experiment results confirm the effectiveness and efficiency of our approach. In particular, \tool (dis)proves all the (non-)constant-time implementations and find new vulnerabilities in open-source libraries Mbed SSL and BearSSL,
and is typically significantly faster than
state-of-the-art tools.

In summary, we make the following main contributions:
\begin{itemize}
  \item We provide a novel synergy of taint analysis and self-composition, improving the efficiency and scalability of verification of constant-time cryptographic implementations;

 \item We develop a fully automated tool to support the verification for production software at the industrial level;

 \item We conduct an extensive evaluation on a large set of real-world programs, and identify new, confirmed timing side-channel vulnerabilities of open-source SSL libraries.   
\end{itemize}

\medskip
\noindent
{\bf Outline.} Section~\ref{sec:pre} presents the background of the work, including the {\while} language and basics of constant-time security. Section~\ref{sec:motivatiingexample} gives motivating examples and an overview of our approach. Section~\ref{sec:method} presents the details of our approaches. Section~\ref{sec:expe} reports the experiment results. Section~\ref{sect:related} discusses the related work. The paper is concluded in Section~\ref{sec:concl}. 
  	 
\section{Preliminaries} \label{sec:pre}
Our constant-time verification tool works over programs in LLVM intermediate representation (IR). There are three main reasons: (1) LLVM IR is a low-level architecture-independent language so verification can be performed on optimized LLVM IR programs, 
(2) it is more convenient to verify and debug LLVM IR programs than binary executables,
and (3) the verified compiler CompCert offers constant-time preserving compilation~\cite{BartheBGHLPT20}
and Binsec/Rel~\cite{DanielBR20,DanielBR23} offers bug-finding and bounded verification for binary executables.
However, for clarity, we use the \while language enriched with arrays, assert/assume statements and procedures, to
define the notion of constant-time security and formalize our verification approach.

\subsection{The \while Language}
{\bf Syntax}. The syntax of the \while language is given in Figure~\ref{fig:syntax}.
A \while program consists of a sequence of procedures, one of which is the {\smain} procedure as the entry of the program.
A procedure contains a sequence of formal arguments and a sequence of
statements followed by a {\sreturn} statement. Note that in our \while language,
the {\sreturn} statement can return a tuple of values which is required for constructing
cross-products. 
For each procedure $f$, we denote by $\Var_f$ the set of  variables used by $f$, including its formal arguments.

\begin{figure}[t]
\small
\setlength\FrameSep{0mm}
\begin{framed}
\[\begin{array}{lrcl}
\mbox{Expressions:} &  e  & ::= & n \mid x\mid e_1 \odot e_2 \mid x[y] \\
\mbox{Statements:} &   p & ::= &\sskip \mid x:=e \mid x[y]:=z \mid \sassert\ e \mid \sassume\ e\mid p;\ p \mid \swhile\ x\ \sdo\ p\ \sod\\
                  &  &  \mid  & \sif\ x\ \sthen\ p_1\ \selse\ p_2\ \sfi \mid 
                  x_1,\cdots,x_m:=f(y_1,\cdots,y_n)  \\
\mbox{Procedures:} &   fn & ::= &\sdef\ f(x_1,\cdots,x_n)\{ p;\ \sreturn~y_1,\cdots, y_m;\}\\
\mbox{Programs:} &   P & ::= &fn^{+}
\end{array}\]
\end{framed}
\vspace{-2mm}
  \caption{The syntax of \while.}\label{fig:syntax}
\end{figure}

Statements include \sskip\ statements, \sassert\ and \sassume\ statements, sequential statements,
\sif-\sthen-\selse\ and \swhile-\sdo\ statements,  assignments, and procedure calls.
As in {\sf ct-verif}~\cite{ABBDE16}, we include the \sassert\ and \sassume\ statements to simplify the reduction to safety checking.  
We assume that each statement is annotated
by a distinct label $\ell$.
Expressions include constants, variables, arithmetic
operations and array accesses. We use $\odot$ to range over
binary operators which are deterministic and side-effect free.
(Unary operators can be defined similarly and are not presented here.)
W.l.o.g., We assume that 
\while programs are given in the single-static assignment (SSA) form, array-read $x[y]$ cannot be used as sub-expressions and all identifies in a program are distinct.

\medskip
\noindent
{\bf Semantics}.
Fix a \while\ program $\Prog$.
Let $\Var$ be the set of variables of $\Prog$,
$\Loc\subseteq\Var\cup(\Var\times \Nn)$ be the set of locations comprising scalar variables and pairs $(x,i)$ of array variables $x$ and indices $i$,
and $\Val$ be the set of possible values of variables.
A \emph{state} $s:\Loc \rightarrow \Val$ is a mapping from
locations $l$ to values $s(l)$, namely, it
maps variables $x$ (resp. array elements $x[i]$) to values
$s(x)$ (resp. $s(x,i)$).
An \emph{initial state} $s_0$ is a mapping that only gives the values of the input variables of $\Prog$, i.e., the parameters of the \smain\ procedure.
The update of a state $s$ is written as $s[l \mapsto n]$.
We define $\bot$ as a distinguished \emph{error state} at which the execution of the program
is disabled. We denote by $s(e)$ the value of the expression $e$ in the state $s$,
i.e., $s(n)=n$, $s(e_1\odot e_2)=s(e_1)\odot s(e_2)$,
and $s(x[y])= s(x,s(y))$. 

A \emph{configuration} $c$ is a pair $\langle s,p\rangle$ consisting of a state $s$ and a statement $p$ to be executed.
An \emph{initial configuration} $c_0$ is a pair $\langle s_0,p\rangle$ such that
$s_0$ is an initial state and $p$ is the statement of the \smain\ procedure excluding the ending \sreturn\ statement.
The semantics of the program $\Prog$ is defined as a transition relation $c{\rightarrow} c'$
between two configurations. We denote by ${\rightarrow}^{\star}$  the reflexive and transitive closure of the relation $\rightarrow$.
The transition relation $c{\rightarrow} c'$  is given in~\Cref{fig:semantics}. 
For the sake of simplicity, array bounds are not checked in our semantics,
and we assume that they are checked by using \sassert\ statements, thus
executions are stuck on the error state when indices are out of the range of the array.

An \emph{execution} $\rho$ of the program $\Prog$ is a sequence of configurations $c_0c_1\cdots c_n$
such that $c_0$ is an initial configuration and $c_i{\rightarrow} c_{i+1}$ for every $0\leq i<n$. 
An execution $\rho$ is \emph{safe} if it does not stick on a configuration $\langle \bot, \cdot\rangle$ with the error state $\bot$,
and is \emph{complete} if it ends with a configuration $\langle \cdot, \sskip\rangle$. 
The program $\Prog$ is \emph{safe} if all the executions are safe.
We remark that in this work, we assume that programs always terminate (i.e., either complete or stick on),
because we focus on cryptographic implementations.
Termination can be checked by tools, e.g.,
CPAchecker~\cite{BeyerK11} and T2~\cite{BrockschmidtCIK16}.

\begin{figure}[t]
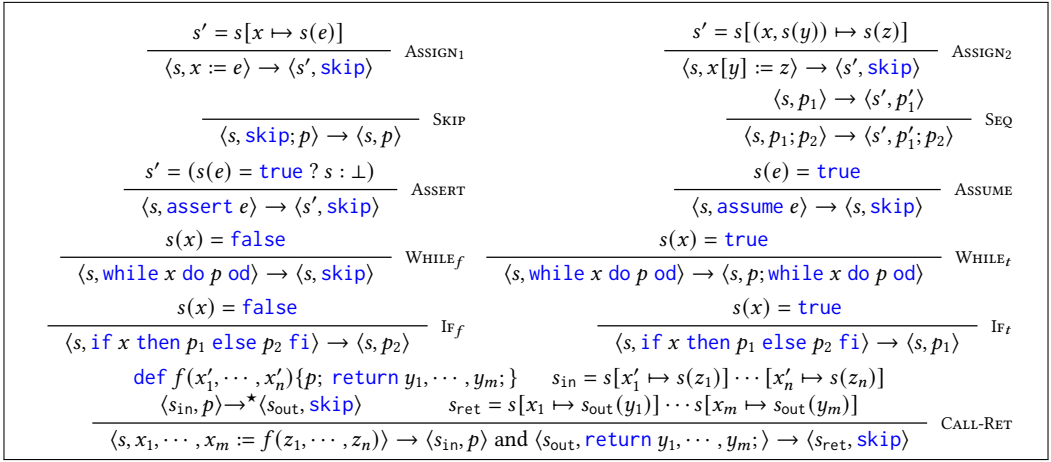

\setlength\FrameSep{1mm}
\begin{framed}
\centering\setlength{\tabcolsep}{4pt}
\renewcommand{\arraystretch}{2}
 \scalebox{0.85}{
 \begin{tabular}{rr}
   $\inference{s'=s[x\mapsto s(e)]} {\langle s, x:=e \rangle \rightarrow \langle s',\sskip\rangle }[{\sc Assign}$_1$]$
 & $\inference{s'=s[(x,s(y))\mapsto s(z)]} {\langle s, x[y]:=z \rangle \rightarrow \langle s',\sskip\rangle }[{\sc Assign}$_2$]$ \\ 
   $\inference{} {\langle s, \sskip; p \rangle \rightarrow \langle s,p \rangle }[{\sc Skip}]$
 & $\inference{\langle s, p_1\rangle \rightarrow \langle s', p_1'\rangle } {\langle s, p_1; p_2 \rangle \rightarrow \langle s', p_1'; p_2 \rangle}[{\sc Seq}]$\\
   $\inference{s'= (s(e)=\true~?~s:\bot)} {\langle s, \sassert\ e \rangle \rightarrow \langle s',\sskip \rangle }[{\sc Assert}]$
 & $\inference{s(e)=\true} {\langle s, \sassume\ e \rangle \rightarrow \langle s, \sskip \rangle }[{\sc Assume}]$\\
   $\inference{s(x)=\false} {\langle s, \swhile\ x\ \sdo\ p\ \sod \rangle \rightarrow \langle s, \sskip \rangle }[{\sc While}$_f$]$
&  $\inference{s(x)=\true} {\langle s, \swhile\ x\ \sdo\ p\ \sod \rangle \rightarrow \langle s, p; \swhile\ x\ \sdo\ p\ \sod \rangle }[{\sc While}$_t$]$\\
$\inference{s(x)=\false} {\langle s, \sif\ x\ \sthen\ p_1\ \selse\ p_2\ \sfi\rangle \rightarrow \langle s, p_2\rangle }[{\sc If}$_f$]$
  &  $\inference{s(x)=\true} {\langle s, \sif\ x\ \sthen\ p_1\ \selse\ p_2\ \sfi\rangle \rightarrow \langle s, p_1\rangle }[{\sc If}$_t$]$ 
  \\ \specialrule{0em}{3pt}{3pt}
  \multicolumn{2}{r}{$\inference{\sdef\ f(x_1',\cdots,x_n')\{ p;\ \sreturn~y_1,\cdots, y_m;\} & s_{\tt in}=s[x_1'\mapsto s(z_1)]\cdots[x_n'\mapsto s(z_n)] \\
\langle s_{\tt in}, p\rangle{\rightarrow}^{\star} \langle s_{\tt out}, \sskip\rangle & s_{\tt ret}=s[x_1\mapsto s_{\tt out}(y_1)]\cdots s[x_m\mapsto s_{\tt out}(y_m)]
} {\langle s, x_1,\cdots,x_m:=f(z_1,\cdots,z_n) \rangle\rightarrow \langle s_{\tt in}, p\rangle\ \mbox{and}\ \langle s_{\tt out}, \sreturn~y_1,\cdots, y_m;\rangle\rightarrow \langle s_{\tt ret}, \sskip\rangle}[{\sc Call-Ret}]$}\\
\end{tabular}}
\end{framed}
\vspace{-2mm}
  \caption{The operational semantics of the \while program.}
  \label{fig:semantics}\vspace{-2mm}
\end{figure}


\subsection{Constant-Time Security}\label{sec:CT}
In practice, execution time variations can create 
timing side-channel in various forms:
(1) unbalanced branching statements may expose the information of the branching condition,
(2) non-constant loops may expose the information of the loop condition,
(3) memory access patterns (cache hits and misses) may expose the information of the memory address
(i.e., indices of arrays in the \while\ programs) accessed in load and store instructions,
(4) time-variant instructions (e.g., integer divisions) in some architectures may expose the information of operands,
and (5) micro-architectural features (e.g., Spectre \cite{KocherHFGGHHLM019} and Meltdown \cite{Lipp0G0HFHMKGYH18}) may break conventional constant-time guarantees.
In this work, we consider the first three timing side-channels 
and use the common leakage model~\cite{ABBDE16,BPT17,DanielBR20,YLW23} to characterize constant-time security.
(The last two timing side-channels are not considered because they are architecture-dependent while we consider LLVM IR which is architecture-independent.) Note that our methodology is generic
and could be adapted to handle the fourth timing side-channel by listing
all the timing-sensitive operations~\cite{ABBDE16}.
Detection, verification and mitigation techniques of the fifth type of timing side-channels 
have been studied (cf.~\cite{CauligiDMBS22} for a survey), most of which assume that
programs are constant-time without micro-architectural features, and thus are orthogonal to our work.

\begin{definition}[Constant-time Leakage Model]\label{def:leakagemodel}
Given a configuration $c=\langle s,p\rangle$ such that $s\neq \bot$, the observation $\obs(c)$   is defined as follows. 
\begin{enumerate}
  \item if $p$ is a branching statement $\sif\ x\ \sthen\ p_1\ \selse\ p_2\ \sfi$, then $\obs(c)=s(x)$, namely, the value of the branching condition $x$ is observable to the adversary;
  \item if $p$ is a loop statement $\swhile\ x\ \sdo\ p'\ \sod$, then $\obs(c)=s(x)$, namely, the value of the loop condition $x$ is observable to the adversary;
  \item if $p$ is an assignment $z:=y[x]$, then $\obs(c)=s(x)$, namely, the value of the index $x$ in the load instruction is observable to the adversary;  
  \item if $p$ is an assignment $y[x]:=z$, then $\obs(c)=s(x)$, namely, the value of the index $x$ in the store instruction is observable to the adversary,
\item if $p$ is an sequential statement $p_1;p_2$, then $\obs(c)=\obs(\langle s,p_1\rangle)$, namely, only the executing instruction $p_1$ is considered (Note that $p_2$ will be considered in a subsequent configuration);
\item otherwise $\obs(c)=\epsilon$, where $\epsilon$ denotes an empty observation.
\end{enumerate}
\end{definition}

For each statement $p$ with label $\ell$ (denoted by $\ell:p$) 
as per Definition~\ref{def:leakagemodel}(1)-(4) where $x$ is the operand or condition, $(\ell,x)$ is called a \emph{potential (timing) side-channel source}.
Intuitively, the value of $x$ at label $\ell$ is observable to the adversary.

An execution $\rho=c_0\cdots c_n$ yields the observation $\obs(\rho)=\obs(c_0)\cdots \obs(c_n)$. Two executions $\rho_1$ and $\rho_2$ are \emph{indistinguishable} (to the adversary with respect to the leakage model $\obs$) if $\obs(\rho_1)=\obs(\rho_2)$.
It is easy to see that two indistinguishable executions $\rho_1$ and $\rho_2$ must have the same control flow, i.e.,
they execute the same conditional branches and iterations of loops, thus the sequences of executed statements are the same. 
This observation is utilized to define cross-product~\cite{ABBDE16}, i.e., self-composition with lockstep execution of both loops and branches,
namely, the copies share the same control flow.

Given a program $\Prog$, we assume that the input variables $\Var^{in}\subseteq \Var_{\smain}$ are partitioned into \emph{public} input variables $\Var^{in}_l$ and \emph{secret} input variables $\Var^{in}_h$. 
These sets are to be annotated by users.
(For the sake of presentation, an input array variable should be annotated by either public or secret, meaning
that all the elements of the array are public or secret.)
In our implementation, we provide API wrappers to precisely
annotate elements of input arrays and fields of
input structures. 
The adversary knows the implementation details of the program and has access to the values of public input variables at runtime, but does not have any direct access to the values of other variables. 
The goal of the adversary is to infer the information of secret input variables
by analyzing observations from executions.

Given a set of variables $X\subseteq \Var$, two states $s_1$ and $s_2$ are \emph{$X$-equivalent}, written as $s_1\simeq_X s_2$,
if for every scale variable $x\in X$, $s_1(x)=s_2(x)$ and for every array variable $x\in X$ and possible index $i\in \Nn$,
$s_1(x,i)=s_2(x,i)$. 
Two configurations $c_1$ and $c_2$ are \emph{$X$-equivalent}, written as $c_1\simeq_X c_2$,
if their states are \emph{$X$-equivalent}. 
For a pair of executions $\rho=c_0c_1\cdots c_n$ and $\rho'=c_0'c_1'\cdots c_{n'}'$, $\rho\simeq_X \rho'$ denotes that for every $0\leq i\leq \min(n,n')$, $c_i \simeq_X c_i'$.

\begin{definition}[Constant-time Security~\cite{ABBDE16}]
A safe program $\Prog$ is \emph{(constant-time) secure} if
for any pair of complete executions $\rho=c_0c_1\cdots c_n$ and $\rho'=c_0'c_1'\cdots c_{n'}'$, 
\begin{center}
 $(c_0\simeq_{\Var^{in}_l} c_0')\Rightarrow \obs(\rho)=\obs(\rho').$   
\end{center}
\end{definition}

Intuitively, the safe program $\Prog$ is secure if, for any pair of complete executions that have the same public input values (i.e., the values of public input variables), their observations are the same, meaning that 
secret inputs are not distinguishable from the observations.
Otherwise, there must exist a side-channel source $(\ell,x)$,
namely, the values of the variable $x$ at label $\ell$ differ between two configurations $c_i$ and $c_i'$ for some $i$. 
Thus, a potential timing side-channel source $(\ell,x)$ does not necessarily leak the secrets (i.e., $x$ is secret-independent), but leaks the secrets when $x$ is secret-dependent.
The security of unsafe programs is undefined, because they are stuck in the error state. 
(Note that the safety of a program can be verified using standard verification techniques and tools, e.g., SMACK~\cite{rakamaric2014smack}.) 

Standard library functions 
{\tt malloc}, {\tt free}, {\tt memcpy} and {\tt memset} may be used in cryptographic implementations.
To handle them, following~\cite{ABBDE16},
we assume that the address and the length used 
in those functions are observable to the adversary,
and the return of {\tt free} is secret-independent.

\section{Motivation and Overview}\label{sec:motivatiingexample}

In this section, we present two motivating examples and an overview of our approach.

\subsection{Motivating Examples}

\begin{figure}[t]
  \centering
\begin{lstlisting}
uint64_t fixfrac(char* frac) {
    uint64_t pow10_LUT[20] = {0x1999999999999999, ...,0x0000000000000000};
    uint64_t pow10_LUT_extra[20] = {0x99, ..., 0x2f};
    uint64_t result = 0;         uint64_t extra = 0;
    for(int i = 0; i < 20; i++) {
        if(frac[i] == '\0') {  break;  }
        uint8_t digit = (frac[i] - (uint8_t) '0');
        result += ((uint64_t) digit) * pow10_LUT[i];
        extra  += ((uint64_t) digit) * pow10_LUT_extra[i];
    }
    ... }
  \end{lstlisting}\vspace{-2mm}
  \caption{Fragment of the function {\tt fixfrac} taken from the libfixedtimefixedpoint library.}\vspace{-2mm}
  \label{fig:example-i32-sub}
\end{figure}

\smallskip
\noindent
{\bf Example 1}.
Figure~\ref{fig:example-i32-sub} shows
a fragment of the function {\tt fixfrac}, a fixed-point numeric operation 
provided by the library  libfixedtimefixedpoint~\cite{7163051}.
Given a digit string {\tt frac} whose length is no more than $20$, it computes a 64-bit number which corresponds to
$${\tt atoi(frac+padding)}/ 10^{20})  * 2^{64}$$
where {\tt frac+padding} is a digit string obtained from {\tt frac} by padding some 0's such that the length is $20$, and ${\tt atoi}$ coverts a digit string into the corresponding integer.
The function {\tt fixfrac} is invoked by the
function {\tt fix\_pow}$(x,y)$ which computes  $x^y$
over the fixed-point numbers $x$ and $y$.

There are five potential side-channel sources in this code snippet, i.e.,
$(6, {\tt frac[i]} == $`$\backslash 0$'$)$, 
$(6,{\tt i})$, $(7,{\tt i})$, $(8,{\tt i})$ and $(9,{\tt i})$.
We can observe that ${\tt i}$ is secret-independent, thus the last four pairs are not side-channel sources. However, it is non-trivial to determine the first potential side-channel source $(6, {\tt frac[i]} == $`$\backslash 0$'$)$ . If it is, 
the information of the secret inputs $x$ and $y$  can be inferred by the adversary via timing side-channels.

This example cannot be proved by
the self-composition based approach {\sf ct-verif}~\cite{ABBDE16}
unless the loop is unrolled or the following loop invariant is added:
\[\exists i_{\max } \cdot 0 \leq i<i_{\max } \leq 20 \wedge \operatorname{frac}\left[i_{\max }\right]==0 \wedge \forall j .0 \leq j \leq i_{\max } \Rightarrow \operatorname{public}(\operatorname{frac}[j]).\]
However, loops may be not statically bounded
and providing such loop invariants manually is non-trivial, which means that in practice, the self-composition based approach is not fully automatic. In contrast, taint analysis is useful here.
Indeed, 
{\tt frac} is secret-independent.

\begin{figure}[t]
  \centering 
\begin{lstlisting}
int crypto_stream_chacha20_ref(unsigned char *c, unsigned long long clen,
                          const unsigned char *n, const unsigned char *k){
    uint32_t ctx[16];
    chacha_keysetup(ctx, k);      // store k from ctx[0] to ctx[11];
    chacha_ivsetup(ctx, n, NULL); // store IV and counter into the rest;
    chacha_encrypt_bytes(ctx, c, c, clen);
}
static void chacha_encrypt_bytes(uint32_t *x, const u8 *m, u8 *c, 
                                 unsigned long long bytes){
    j12 = x[12];
    if (!j12) {...}
}

  \end{lstlisting}\vspace{-3mm}
  \caption{Simplified fragment of the function {\tt crypto\_stream\_chacha20\_ref} taken from 
  the libsodium library.}
  \label{fig:example-part-chacha20}\vspace{-2mm}
\end{figure}

\smallskip
\noindent
{\bf Example 2}.
Figure~\ref{fig:example-part-chacha20}
shows a simplified fragment of the function 
{\tt crypto\_stream\_chacha20\_ref} taken from the libsodium library, a portable, cross-compilable and installable fork of NaCl with an extended API to improve usability. 
This function implements the ChaCha20 stream cipher~\cite{bernstein2008chacha}.

In the fragment of {\tt crypto\_stream\_chacha20\_ref},
variable {\tt c} points to a plaintext to be encrypted,
{\tt clen} is the length of the plaintext, 
{\tt n} points to an initialization vector, 
and {\tt k} points to a private key.
The key {\tt k} is stored at the first 12 positions of the buffer {\tt ctx},  the initialization vector and a counter (initialized as NULL) are stored at the rest 4 positions of the buffer {\tt ctx}.

There is one potential side-channel sources in this simplified fragment, i.e., $(11,{\tt !j12})$.
It is non-trivial to determine 
this potential side-channel source, as the value of {\tt j12} is {\tt x[12]} while the buffer {\tt x} contains the secret key {\tt k}. If it is, the information of the secret key {\tt k} can be inferred by the adversary via timing side-channels.

This example cannot be proved by an index-insensitive taint analysis, because  
{\tt ctx} contains both secret-dependent and secret-independent contents, namely, {\tt ctx[0--11]} and
{\tt ctx[12--15]}.
Any index-insensitive taint analysis will conservatively taint the whole buffer.
In contrast, an index-sensitive taint analysis
(e.g., our precise taint analysis)  
would work in this example.

These examples reveal that  
static analysis (e.g., taint analysis) and self-composition based approaches may have complementary strengths even without efficiency considerations. Our method precisely takes advantage of their respective strength for which we provide an overview below. 


 \begin{figure}[t]
	\centering
	\includegraphics[width=1\columnwidth]{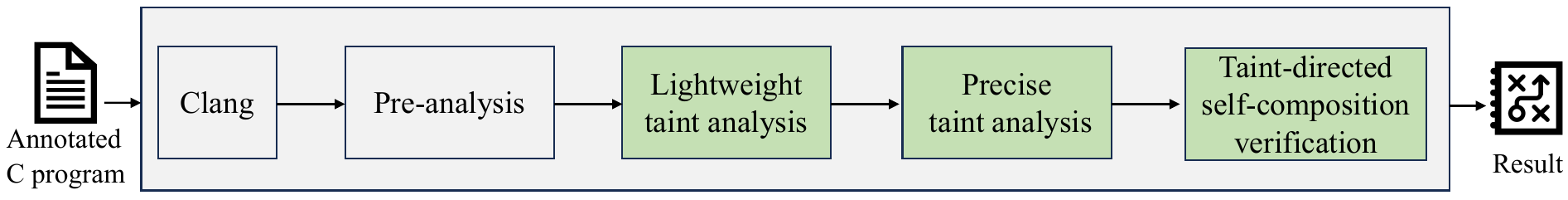}
	\vspace{-6mm}
       \caption{Overview of our approach}
\label{fig:overview}\vspace{-2mm}
\end{figure}

\subsection{Approach Overview}
An overview of our approach is shown in Figure~\ref{fig:overview}.  
In general, for a given annotated C program, the tool either outputs \emph{proved}, suggesting that the program is secure, or outputs
(potential) side-channel sources and corresponding execution traces for vulnerability localization and repair.

Our approach works as follows.
First, the input program is translated into LLVM intermediate representation (IR) with annotations using Clang, the front-end of LLVM. Second, we compute the call graph, interprocedural control-flow graph, points-to information and definition-use chains required by the subsequent steps via a pre-analysis.
Third, a lightweight taint analysis (cf.\ Section~\ref{sec:1sttaint}) is performed by leveraging the inter-procedural, finite, distributive, subset (IFDS)
framework~\cite{RHS95}. Often it is able to determine a large number of potential side-channel sources, leaving few potential side-channel sources unresolved. 
Finally, we resort to a precise taint analysis (cf.\ Section~\ref{sect:2ndtaint}) and a taint-directed self-composition verification (cf.\ Section~\ref{sec:cpverif}) to resolve the left-over potential side-channel sources.
The precise taint analysis reduces the flow-, context-, path-, field- and index-sensitive taint analysis problem to checking safety properties of a cross-product of the given program and its Boolean abstraction which tracks the required information flow from the secrets. 
The taint-directed self-composition consists of two copies of the original program which is a new variant of self-composition.  
Remarkably, the taint information is utilized in 
both cross-product constructions to simplify the resulting program and reduce
the cost of safety checks.  
By combining lightweight taint analysis and heavyweight safety verification, the overall approach brings the best of three worlds: efficiency, soundness and theoretical completeness. 

Consider the motivating examples.
The lightweight taint analysis is able to prove that all the five potential side-channel sources in Example 1 actually do not leak secrets, so the next two  
steps are not needed. In Example 2,
the lightweight taint analysis fails to determine
the potential side-channel source $(11,{\tt !j12})$, thus the subsequent analyses have to check  $(11,{\tt !j12})$, which can be resolved by the precise taint analysis.  (The final step is thus not needed.)
We remark that the precise taint analysis may fail to prove some constant-time implementations meaning that it is sound but incomplete. For instance, when the secret $k$ is involved in the computation of a potential timing side-channel source $x$ (e.g., 
$x=k\oplus p\oplus k$ where $\oplus$ is Exclusive-OR), and the value of $x$ is independent upon the secret $k$, the precise taint analysis will raise a false positive.

\section{Methodology}
\label{sec:method}

In this section, we present the details
of the three key components, i.e.,  
lightweight taint analysis,
precise taint analysis and taint-directed self-composition.

\subsection{Lightweight Taint Analysis}\label{sec:1sttaint}
In this subsection, we present a lightweight taint analysis  
which is designed to be flow-, field- and context-sensitive, but path-and index-insensitive, for a balance of efficiency and precision. 
Often it is able to prove that a large number of potential side-channel sources do not leak secrets, leaving few potential side-channel sources unresolved.

\smallskip
\noindent
{\bf Taint source}. 
Fix a safe program $\Prog$. The taint source  
is the set $\Var^{in}_h$ of its secret input variables, each of which
is a taint fact. We remark that although our implementation
supports the element-wise annotation of input array variables, 
this taint analysis is index-insensitive.
Thus, if any element of an input array variable is annotated by secret,
the array variable is regarded as secret, i.e., all the elements of the array are tainted.

\begin{figure}[t]
\renewcommand{\arraystretch}{2.2}
\setlength\FrameSep{1mm}
\begin{framed}
\centering
 \scalebox{0.85}{  
\begin{tabular}{rr}
  $\inference{T'=\big(y\in T\ ?\ T\cup \{x\}:T\setminus\{x\}\big)} {x:=y[z] \vdash T \hookrightarrow T'}[T-{\sc Assign}$_1$]$ &
  $\inference{T'=\big(\var(e)\cap T\neq \emptyset\ ?\ T\cup \{x\}:T\setminus 
  \{x\}\big)} {x:=e \vdash T \hookrightarrow T'}[T-{\sc Assign}$_3$]$ \\
 $\inference{T'=\big(z\in T\ ?\ T\cup \{x\}:T\big)} {x[y]:=z \vdash T \hookrightarrow T'}[T-{\sc Assign}$_2$]$&
 $\inference{p \mbox{ is } \sskip \mbox{ or } \sassert\ e \mbox{ or } \sassume\ e} {p \vdash T \hookrightarrow T}[T-{\sc Identity}]$    \\
 $\inference{p_1\vdash T \hookrightarrow T_1 & p_2\vdash T_1 \hookrightarrow T'} {p_1;p_2\vdash T \hookrightarrow T'}[T-{\sc Seq}]$ &
 $\inference{T'=\lfp(p,T)} {\swhile\ x\ \sdo\ p\ \sod \vdash T \hookrightarrow T'}[T-{\sc While}]$    \\
 & $\inference{p_1\vdash T \hookrightarrow T_1 & p_2\vdash T \hookrightarrow T_2} { \sif\ x\ \sthen\ p_1\ \selse\ p_2\ \sfi \vdash T \hookrightarrow T_1\cup T_2}[T-{\sc If}]$ 
\\	\specialrule{0em}{3pt}{3pt}
   \multicolumn{2}{r}{$\inference{\sdef\ f(x_1',\cdots,x_n')\{ p;\ \sreturn~y_1,\cdots, y_m;\} &
  T_{\tt in}=(T\setminus \Var_g)\cup \{x_i'\mid 1\leq i\leq n\wedge z_i\in T\} \\ 
   \mbox{caller}=g \quad
  p \vdash T_{\tt in} \hookrightarrow T_{\tt out}\quad
  T'= (T\setminus\{x_1,\cdots,x_m\}\cup T_{\tt out}\setminus\Var_f\cup\{x_i\mid 1\leq i\leq m\wedge y_i\in T_{\tt out}\}} {x_1,\cdots,x_m:=f(z_1,\cdots,z_n)  \vdash T \hookrightarrow T'}[{\sc T-Call-Ret}]$}\\
  \end{tabular}}  
  \end{framed}\vspace{-2mm}
  \caption{Taint inference rules for the {\while} language.}
  \label{fig:transfer}\vspace{-3mm}
\end{figure}

\smallskip
\noindent
{\bf Taint inference rule}. The taint inference rule is given
by the  transfer function of the form
\[p \vdash T \hookrightarrow T',\]
where $p$ is a statement, $T$ and $T'$ are sets of taint facts.
The transfer function $p \vdash T \hookrightarrow T'$ means that 
the execution of the statement $p$ with the set of taint facts
$T$ results in the set of taint facts $T'$.

The taint inference rules of the {\while} language are given in Figure~\ref{fig:transfer}, where
$\var(e)$ denotes the set of variables involved in the expression $e$, and $\lfp(p,T)$ is recursively defined as follows:
\[\lfp(p,T)=\left \{
\begin{array}{ll}
    T, &  \mbox{if}\ {p  \vdash T \hookrightarrow T};\\
     T\cup\lfp(p, T'), & \mbox{otherwise, where}\ {p  \vdash T \hookrightarrow T'}.
\end{array}\right.\]

These rules are standard, which, intuitively, propagate taints from the right-hand side variables to the left-hand side variable. 
For example, rule [T-{\sc Assign}$_1$] expresses that if the array $y$ is tainted, then the loaded array element is tainted. Rule [T-{\sc Assign}$_2$] expresses that if a tainted value is stored in an array, then the array is tainted. (Recall that our taint analysis is index-insensitive.) Rule [T-{\sc Assign}$_3$] expresses that if any involved variable of an expression $e$ is tainted, then the result of the expression $e$ is also tainted. Note that if the left-hand side is a scalar variable (i.e., rules [T-{\sc Assign}$_1$] and [T-{\sc Assign}$_3$]), we perform a strong update.
Rule [T-{\sc While}] computes the least fixed point by applying the operator $\lfp$ which always terminates, because
the sequence of sets of taint facts during $\lfp(p,T)$ is ascending w.r.t. the order $\subseteq$. 
Rule [{\sc T-Call-Ret}] is much involved. The set $T_{\tt in}$ of input taint facts at the call-site is obtained by
passing actual arguments of the caller $g$ to the formal parameters of the callee $f$ after filtering out the taint facts of
the local variables $\Var_g$ of the caller $g$. The body $p$ of the callee is analyzed
using the set $T_{\tt in}$ of input taint facts, leading to the set $T_{\tt out}$ of output taint facts.
The set $T_{\tt out}$ of output taint facts is merged with the set $T$ of taint facts at return-site after filtering out the taint facts of
the local variables $\Var_f$ of the callee $f$ and the actual return variables are updated accordingly.

\smallskip
\noindent
{\bf IFDS-based taint analysis}.
We leverage the inter-procedural, finite, distributive, subset (IFDS) framework to implement the lightweight taint analysis.
The time complexity and space complexity of the vanilla IFDS algorithm are ${\bf O}(|E|\cdot |D|^3)$ and  ${\bf O}(|E|\cdot|D|)$, respectively,
where $|E|$ is the number of edges in the interprocedural control-flow graph and the size $|D|$ of the domain is the number of all possible taint facts, i.e., $|\Var|$. The time complexity will increase sharply with $|\Var|$ in practice~\cite{li2021scaling} which
is significant for some cryptographic implementations due to the following reasons. 

On the one hand, typically, the input of a cryptographic algorithm consists of a key and plaintext; the key is secret and thus tainted,
and the key and plaintext are tightly coupled in the computation.
For instance, the AES algorithm has multiple modes, with the smallest key size being 128 bits (an array with 16 elements) and the smallest encryption process being 10 rounds each of which has four transformations. On the other hand, the SSA form introduces a number of temporary variables. Consequently, a large number of taint facts are propagated during the taint analysis.  
To mitigate this issue, inspired by the sparse data-flow analysis~\cite{OhHLLY12}, we improve the classic IFDS framework
by directly propagating taint facts of scalar variables
via data flow instead of control flow. More specifically, 
if a scalar variable $x$ is tainted, this taint fact is directly propagated to the statements where $x$ is used, using 
the def-use chains. It avoids the propagation of taint facts
of scalar variables for many statements, hence improving efficiency in practice.

Hereafter, for every label $\ell$ of a statement $p$, we denote by
$T_\ell$ the set of taint facts at label $\ell$ (i.e., before
the execution of the statement $p$), obtained by applying
the IFDS-based taint analysis. 
Since the IFDS-based taint analysis is sound, it is straightforward to have that 

\begin{lemma}\label{lem:lightweightTA}
For any potential side-channel source $(\ell,x)$ and pair $(c_0,c_0')$ of initial configurations such that $(c_0\simeq_{\Var^{in}_l} c_0')$, 
if $x$ is not tainted at the label $\ell$, i.e.,
$x\not\in T_\ell$, then the values of $x$ are the same at the label $\ell$ in any pair of complete executions $\rho=c_0c_1\cdots c_n$ and $\rho'=c_0'c_1'\cdots c_{n'}'$.
\end{lemma}

\begin{proof}[Proof sketch]
For a pair of complete executions $\rho=c_0c_1\cdots c_n$ and $\rho'=c_0'c_1'\cdots c_{n'}'$ with $c_0\simeq_{\Var^{in}_l} c_0'$, suppose that the values of $x$ are different at the label $\ell$. It suffices to show that $x\in T_\ell$, i.e., $x$ is tainted at $\ell$. 
 
Note that the values of public input variables are the same in 
the initial configurations $c_0$ and $c_0'$. We then can deduce that the value of the variable $x$ at the label  $\ell$ depends upon some secret input variables from $\Var^{in}_h$. 
By the soundness of the IFDS framework~\cite{RHS95} and sparse static analysis~\cite{OhHLLY12}, $x$ must be tainted at  $\ell$. 
\end{proof}

Obviously, if for each potential side-channel source $(\ell,x)$,
the variable $x$ is not tainted at the label $\ell$, we can deduce that the safe program is constant-time secure.
 

\subsection{Precise Taint Analysis via Taint-directed Semi-cross-product} \label{sect:2ndtaint}
While the lightweight taint analysis
is often effective in ruling out a large number of potential side-channel sources to be genuinely vulnerable, 
some can still not be determined 
due to the over-approximation nature of the static analysis, e.g., index-insensitive.
In this subsection,  
we propose a precise taint analysis, which reduces the flow-, context-, path-, field- and index-sensitive taint analysis problem to checking safety properties of a novel cross-product, called \emph{taint-directed semi-cross-product}. We shall first explain the intuition and then present
the formal construction.  

\smallskip
\noindent
{\bf Intuition}. Given a program $\Prog$, we construct a 
semi-cross-product $\Prog'$ of the program $\Prog$ and its Boolean abstraction.  
Here, `semi' means that
one copy in the cross-product $\Prog'$ is a Boolean abstraction of $\Prog$ instead of the original one; `cross' means that $\Prog'$ shares the same control flow of $\Prog$ and executes statements of two copies in a lockstep manner. 

The Boolean abstraction has 
\begin{itemize}
    \item a Boolean variable $b_x$ for each scalar variable $x$ in  $\Prog$ such that $b_x=1$ iff $x$ is tainted. 
    \item  a Boolean array $b_x$ for each array variable $x$ in $\Prog$ such that $b_x[i]=1$ iff $x[i]$ is tainted.  
\end{itemize}

The precise taint analysis is to determine whether a scalar variable $x$ (resp. 
an array element $x[i]$) is tainted or not. For this purpose, it suffices to check
whether there exist inputs to $\Prog'$ such that $b_x$ (resp. $b_x[i]$) is $1$, which is a standard safety verification problem.

To reduce the cost of safety verification, we incorporate the results from the lightweight taint analysis into 
the semi-cross-product $\Prog'$ based on the following observation. The lightweight taint analysis is conservative (i.e., intuitively it may overly taint), consequently, a more precise taint analysis would not taint the variables 
that have not been tainted by the lightweight taint analysis. Hence
if a variable $x$ has not been tainted by the lightweight taint analysis, 
its Boolean abstraction $b_x$ in the semi-cross-product  $\Prog'$ will always be $0$.
As a result, when verifying $\Prog'$, it suffices to focus exclusively on the variables that have been tainted by the lightweight taint analysis. 
For the variables that have not been tainted by the lightweight taint analysis, 
we can simply assign $0$ to them, which can further improve the efficiency of safety verification.

\smallskip
\noindent
{\bf Product construction}.
The semi-cross-product $\Prog'$
for a given program $\Prog$ is constructed by iteratively applying
the function $\pi(\cdot)$ (given in Fig.~\ref{fig:semi-product}) to each procedure of $\Prog$. For each procedure,  formal parameters and return variables are duplicated by the function $\pi(\cdot)$ using their Boolean abstractions, and the procedure body $p$ is replaced by 
$\pi(p)$. 
Moreover, the Boolean abstractions of public and secret inputs are respectively initialized by $0$ and $1$ at the beginning of the \smain\ procedure.

During the construction, $\xi(e)$ is used to generate the Boolean abstraction of the expression $e$ that computes the taint value of the result of the expression $e$. 
Specifically, 
if $e$ is a constant $n$, $\xi(e)$ is the Boolean constant $0$;
if $e$ a variable $x$, $\xi(e)$ is the Boolean abstraction of $x$ (i.e., $b_x$);
and if $e$ is a compound expression $e_1 \odot e_2$, $\xi(e)$ is 
the disjunction $\xi(e_1)\vee \xi(e_2)$ of taint values of sub-expressions, namely, the value of $e$ is tainted if some variable used in $e$ is tainted.

\begin{figure}[t]
\centering\setlength{\tabcolsep}{7pt}
\setlength\FrameSep{1mm}
\begin{framed}
\renewcommand{\arraystretch}{1.1}
 \scalebox{0.85}{
\noindent \begin{tabular}{lll}
$\xi(n)\triangleq 0$ & \qquad\qquad\qquad\qquad\qquad
$\xi(x)\triangleq b_x$\qquad\qquad\qquad\qquad  &
$\xi(e_1 \odot e_2)\triangleq \xi(e_1)\vee \xi(e_2)$  \\
\multicolumn{2}{l}{$\pi(p)\triangleq p$ if $p$ is $\sskip \mbox{ or } \sassert\ e \mbox{ or } \sassume\ e$} &  $\pi(p_1;\ p_2)\triangleq \pi(p_1);\ \pi(p_2)$ \\
\multicolumn{3}{l}{$\pi(\ell:\ x:=e)\triangleq 
\left\{\begin{array}{ll}
x:=e;\ b_x:=0 & \mbox{if } \var(e)\cap T_{\ell}=\emptyset  \\
x:=e;\ b_x:= \xi(e) &  \mbox{otherwise}
\end{array}  \right.$}\\
\multicolumn{3}{l}{$\pi(\ell:\ x:=y[z])\triangleq 
\left\{\begin{array}{ll}
\Guard_\ell(z);\ x:=y[z];\ b_x:=0 & \mbox{if } y \not\in T_{\ell} \\
\Guard_\ell(z);\ x:=y[z];\ b_x:=b_y[z] &  \mbox{otherwise}\\
\end{array}  \right.$} \\ 
\multicolumn{3}{l}{$\pi(\ell:\ x[y]:=z)\triangleq 
\left\{\begin{array}{ll}
\Guard_\ell(y);\ x[y]:=z;\ b_x[y]:=0 & \mbox{if }  z\not\in T_\ell  \\
\Guard_\ell(y);\ x[y]:=z;\ b_x[y]:=b_z&  \mbox{otherwise}
\end{array}  \right.$} \\ 
\multicolumn{3}{l}{$\pi(\ell:\ \swhile\ x\ \sdo\ p\ \sod)\triangleq \  
\swhile\ x @{\tt INV} \ \sdo\ \Guard_{\bbegin(\ell)}(x);\ \pi(p)\  \sod; \Guard_{\eexit(\ell)}(x)$} \\
\multicolumn{3}{l}{$\pi(\ell:\ \sif\ x\ \sthen\ p_1\ \selse\ p_2\ \sfi)\triangleq \Guard_{\ell}(x);\ \sif\ x\ \sthen\ \pi(p_1)\ \selse\  \pi(p_2)\ \sfi$}\\
\multicolumn{3}{l}{$\pi(x_1,\cdots,x_m:=f(y_1,\cdots,y_n))\triangleq x_1,b_{x_1},\cdots,x_m,b_{x_m}:=f(y_1,b_{y_1},\cdots,y_n,b_{y_n})$}\\
\multicolumn{3}{l}{$\pi(\sdef\ f(x_1,\cdots,x_n)\{ p;\ \sreturn~y_1,\cdots, y_m;\})= \sdef\ f(x_1,b_{x_1}\cdots,x_n,b_{x_n})\{ \pi(p);\ \sreturn~y_1,b_{y_1},\cdots, y_m,b_{y_m};\}$}\\ 
\multicolumn{3}{l}{$\Guard_\ell(x)\triangleq (x\in T_\ell\ ?\ \sassert\ \neg b_x : \sskip)$} \\
\end{tabular}}
\end{framed} 
  \caption{The taint-directed semi-cross-product of the \while programs, where $\bbegin(\ell)$ 
  and $\eexit(\ell)$ denote the labels of the beginning and exit of the loop body of the loop at the label $\ell$.}
  \label{fig:semi-product}
\end{figure}

For each statement $p$ with label $\ell$, $\pi(p)$ produces a new statement. 
If $p$ is a \sskip\ or \sassert\ or \sassume\ statement, $\pi(p)$ gives $p$ itself, namely, it does not have a Boolean abstraction.
If $p$ is a sequential statement $p_1;\ p_2$, $\pi(p)$ gives
the sequential statement $\pi(p_1);\ \pi(p_2)$ by recursively applying
the function $\pi(\cdot)$. If $p$ is a standard assignment $x:=e$, $\pi(p)$ is defined to track the information flow from the operands of the expression $e$ to the Boolean abstraction $b_x$ of $x$ if some operand of $e$ has been tainted by the lightweight taint analysis, otherwise $0$.

If $p$ is a load statement $x:=y[z]$, $\pi(p)$ gives a sequential 
statement $\Guard_\ell(z);\ x:=y[z];\ b_x:=0$ if $y$ has not been tainted (i.e., $y\not\in T_{\ell}$). Otherwise, $\Guard_\ell(z);\ x:=y[z];\ b_x:=b_y[z]$ is generated meaning that $b_x$ is the same as $b_y[z]$.
The auxiliary function $\Guard_\ell(z)$ is used to generate an \sassert\ statement to resolve the taint status of the variable $z$ if $z$ has been tainted, because $(\ell,z)$
is a potential side-channel source. If $z$ has not been tainted, the assertion is avoided by replacing it with a \sskip\ statement which can be further removed from the program $\Prog'$.
A store statement $x[y]:=z$ is handled similarly,
except that 
$b_y=0$ is checked by inserting an \sassert\ statement if $y$ has been tainted (i.e., $y\in T_\ell$)
and $b_x[y]$ is updated accordingly.

If $p$ is a \swhile-\sdo\ statement $\swhile\ x\ \sdo\ p\ \sod$,
$\pi(p)$ inserts an \sassert\ statement at the beginning (resp. exit) of the loop body to ensure that $b_x=0$ if $x$ has been tainted at the beginning (resp. exit) of the loop body. Furthermore, to facilitate safety verification, loop invariants {\tt INV} are inserted. Currently, we use the following heuristic strategy to generate loop invariants and leave the generation of more effective loop invariants as future work. For each variable $y$ defined in
the loop body, if it is used to compute the loop condition $x$, then the predicate $\neg b_y$ is added as a loop invariant. However, such loop invariants may be invalid in practice. Loop invariants are checked and invalid ones are removed during safety verification.

If $p$ is an \sif-\sthen-\selse\ statement $\sif\ x\ \sthen\ p_1\ \selse\ p_2\ \sfi$, similarly, $\pi(p)$ inserts an \sassert\ statement to resolve the taint status of the branching condition $x$ if it has been
tainted. 

For every potential side-channel source $(\ell,x)$
such that $x\in T_\ell$, the program $\Prog'$ must have
one or two \sassert\ statements introduced by \Guard which checks
if $b_x=0$. 
By applying a sound safety verifier, 
we can check if such \sassert\ statements are valid or not.
If they are valid, we can deduce that $(\ell,x)$ is not a side-channel source, thus 
the variable $x$ can be removed from the set of taint facts $T_\ell$.
We denote by $T_\ell'$ the resulting set of taint facts.
It is trivial to see that Lemma~\ref{lem:lightweightTA} still holds
when $T_\ell$ is updated by $T_\ell'$ for every potential side-channel source $(\ell,x)$.

\begin{lemma}\label{lem:preciseTA}
For any potential side-channel source $(\ell,x)$ and pair $(c_0,c_0')$ of initial configurations such that $(c_0\simeq_{\Var^{in}_l} c_0')$,  
if $x\not\in T_\ell'$, then the values of $x$ are the same at the label $\ell$ in any pair of complete executions $\rho=c_0c_1\cdots c_n$ and $\rho'=c_0'c_1'\cdots c_{n'}'$.
\end{lemma}

\begin{proof}[Proof sketch]
 Suppose the values of $x$ are different at the label $\ell$ in a pair of complete executions $\rho=c_0c_1\cdots c_n$ and $\rho'=c_0'c_1'\cdots c_{n'}'$ with $c_0\simeq_{\Var^{in}_l} c_0'$. Let $b_x$ be the Boolean abstraction of $x$.
We show that the assert statement $\sassert\ \neg b_x$
for the potential side-channel source $(\ell,x)$ is not valid. 

As the values of $x$ are different at $\ell$ in $\rho$ and $\rho'$, and
the values of public input variables are the same in 
the initial configurations $c_0$ and $c_0'$, we can deduce that the value of $x$ at $\ell$ depends upon some secret input variables from $\Var^{in}_h$. 
Since the Boolean abstractions of all the secret input variables
are initialized by $1$, the Boolean abstraction of each internal variable is set to $0$ only if the internal variable is independent of secret input variables, and 
Boolean abstractions can only be disjunction, we conclude that $b_x$
is $1$ when the semi-cross-product $\Prog'$ takes the inputs from $c_0$ or $c_0'$. 
\end{proof}


\subsection{Taint-directed Self-composition}\label{sec:cpverif}
The precise taint analysis can resolve potential side-channel sources that were left by the lightweight taint analysis, but may still fail on some potential side-channel sources. Indeed, it cannot determine genuine side-channel sources that would leak secrets.
Thus, in this subsection, we 
propose a taint-directed self-composition which 
improves the original construction~\cite{ABBDE16} by incorporating the taint information. Our construction simplifies the self-composed programs and reduces the number of safety checks. 
We first describe the intuition and then present the formal construction.

\smallskip
\noindent
{\bf Intuition}. Given a program $\Prog$, we construct a 
cross-product $\hat{\Prog}$ 
of the program $\Prog$ comprising two copies of $\Prog$
which share the same control flow.
One copy is the \emph{original} program, and the other copy is referred to as the \emph{shadow} one which has a shadow variable $\hat{x}$ for each variable $x$ in $\Prog$.
Typically we check whether the variable $x$ and 
its shadow $\hat{x}$ have the same values when the original and shadow counterparts of $\hat{\Prog}$ are provided with the same
public inputs but different secret inputs.

To reduce the cost of safety verification, we also incorporate the results from two taint analyses into $\hat{\Prog}$ based on the following observation. If a variable $x$ has not been tainted, 
$x$ and 
its shadow $\hat{x}$
in the product $\hat{\Prog}$ must have the same values when the public inputs of the two copies are the same.
Therefore, when verifying $\hat{\Prog}$, it suffices to focus exclusively on the variables that are still tainted after two taint analyses. Moreover,
for those variables that have not been tainted, 
we can simply assign the value of the original variable to the shadow one instead of copying the computation, which can further improve the efficiency of safety verification.

\smallskip
\noindent
{\bf Product construction}. 
The cross-product $\hat{\Prog}$
for a given program $\Prog$ is constructed by iteratively applying
the function $\Pi(\cdot)$ (given in Fig.~\ref{fig:product}) to each procedure of $\Prog$. 
Moreover, the shadow counterparts of public inputs are initialized by their original ones at the beginning of the \smain\ procedure.
For each procedure, $\Pi(\cdot)$ is defined similar to
the function $\pi(\cdot)$ in semi-cross-product,
where the Boolean abstractions are replaced
by the shadow counterparts.

\begin{figure}[t]
\centering\setlength{\tabcolsep}{7pt}
\setlength\FrameSep{1mm}
\begin{framed}
\renewcommand{\arraystretch}{1.1}
 \scalebox{0.85}{
\noindent \begin{tabular}{lll}
$\Xi(n)\triangleq n$ & \qquad\qquad\qquad\qquad\qquad
$\Xi(x)\triangleq \hat{x}$\qquad\qquad\qquad\qquad  &
$\Xi(e_1 \odot e_2)\triangleq \Xi(e_1)\vee \Xi(e_2)$  \\
\multicolumn{2}{l}{$\Pi(p)\triangleq p$ if $p$ is $\sskip \mbox{ or } \sassert\ e \mbox{ or } \sassume\ e$} &  $\Pi(p_1;\ p_2)\triangleq \Pi(p_1);\ \Pi(p_2)$ \\
\multicolumn{3}{l}{$\Pi(\ell:\ x:=y[z])\triangleq 
\left\{\begin{array}{ll}
\Guard_\ell'(z);\ x:=y[z];\ \hat{x}:=x & \mbox{if } y \not\in T'_{\ell} \\
\Guard_\ell'(z);\ x:=y[z];\ \hat{x}:=\hat{y}[\hat{z}] &  \mbox{otherwise}\\
\end{array}  \right.$} \\ 
\multicolumn{3}{l}{$\Pi(\ell:\ x[y]:=z)\triangleq 
\left\{\begin{array}{ll}
\Guard_\ell'(y);\ x[y]:=z;\ \hat{x}[\hat{y}]:=x[y] & \mbox{if } z\not\in T'_\ell  \\
\Guard_\ell'(y);\ x[y]:=z;\ \hat{x}[\hat{y}]:=\hat{z}&  \mbox{otherwise}
\end{array}  \right.$} \\ 
\multicolumn{3}{l}{$\Pi(\ell:\ x:=e)\triangleq 
\left\{\begin{array}{ll}
x:=e;\ \hat{x}:=x & \mbox{if } \var(e)\cap  T_{\ell}'=\emptyset  \\
x:=e;\  \hat{x}:= \Xi(e) &  \mbox{otherwise}
\end{array}  \right.$}\\
\multicolumn{3}{l}{$\Pi(\ell:\ \swhile\ x\ \sdo\ p\ \sod)\triangleq \  
\swhile\ x @{\tt INV'} \ \sdo\ \Guard'_{\bbegin(\ell)}(x);\ \Pi(p)\  \sod; \Guard'_{\eexit(\ell)}(x)$} \\
\multicolumn{3}{l}{$\Pi(\ell:\ \sif\ x\ \sthen\ p_1\ \selse\ p_2\ \sfi)\triangleq \Guard'_{\ell}(x);\ \sif\ x\ \sthen\ \Pi(p_1)\ \selse\  \Pi(p_2)\ \sfi$}\\
\multicolumn{3}{l}{$\Pi(x_1,\cdots,x_m:=f(y_1,\cdots,y_n))\triangleq x_1,\hat{x}_1,\cdots,x_m,\hat{x}_m:=f(y_1,\hat{y}_1,\cdots,y_n,\hat{y}_n)$}\\
\multicolumn{3}{l}{$\Pi(\sdef\ f(x_1,\cdots,x_n)\{ p;\ \sreturn~y_1,\cdots, y_m;\})= \sdef\ f(x_1,\hat{x}_1\cdots,x_n,\hat{x}_n)\{ \Pi(p);\ \sreturn~y_1,\hat{y}_1,\cdots, y_m,\hat{y}_m;\}$}\\ 
\multicolumn{3}{l}{$\Guard'_\ell(x)\triangleq (x\in T_\ell'\ ?\ \sassert\ x=\hat{x} : \sskip)$} \\
\end{tabular}}
\end{framed} \vspace{-2mm}
  \caption{The taint-directed cross-product of the \while programs.} 
  \label{fig:product}\vspace{-2mm}
\end{figure}

More specifically, the auxiliary function $\Xi(e)$
used in defining $\Pi(\ell:\ x:=e)$ now generates a shadow expression that computes the value of expression $e$ over shadow variables. 
Concretely, if $e$ is a constant $n$, $\Xi(e)$ is the constant $n$; 
if $e$ is a variable $x$, $\Xi(e)$
is the shadow variable $\hat{x}$ of $x$;
and if $e$ is a compound expression $e_1 \odot e_2$, $\Xi(e)$ is $\Xi(e_1)\odot \Xi(e_2)$. 

The auxiliary function $\Guard'_\ell(x)$ 
for each left potential side-channel source $(\ell,x)$ generates an \sassert\ statement to ensure that the variable $x$ and its shadow counterpart have
the same value (i.e., $x=\hat{x}$)
if $x$ has been tainted. It allows us to resolve the taint status of $x$.

The duplicated Boolean counterpart $b_x:=0$ (resp.  
$b_x[y]:=0$) is replaced by $\hat{x}:=x$ (resp. $\hat{x}[\hat{y}]:=x[y]$) in $\Pi(\ell:\ x:=y[z])$
(resp. $\Pi(\ell:\ x[y]:=z)$) if $x$ was untainted,
otherwise the right-hand side is copied using
the corresponding shadow counterparts.

For each $\swhile\ x\ \sdo\ p\ \sod$,
$\Pi(p)$ inserts loop invariants {\tt INV}. 
Following~\cite{ABBDE16}, for each variable $y$ defined in
the loop body, if it is used to compute the loop condition $x$, then the predicate $y=\hat{y}$ is added as a loop invariant. Similar to
the precise taint analysis, such loop invariants
are checked and invalid ones are removed during
safety verification.

For every potential side-channel source $(\ell,x)$
such that $x\in T_\ell'$, the program $\hat{\Prog}$ must have
one or two \sassert\ statements introduced by $\Guard'$ which checks
if $x=\hat{x}$. 
By invoking safety verifiers, 
we can check if such \sassert\ statements are valid or not.
If they are valid, we can deduce that $(\ell,x)$ is not a side-channel source, thus 
the variable $x$ can be removed from the set of taint facts $T_\ell'$.
We denote by $\hat{T}_\ell$ the resulting set of taint facts.
It is trivial to see that Lemma~\ref{lem:lightweightTA} still holds
when $T_\ell$ is updated by $\hat{T}_\ell$ for every potential side-channel source $(\ell,x)$.

\begin{lemma}\label{lem:product}
For any potential side-channel source $(\ell,x)$ and pairs $(c_0,c_0')$ of initial configurations such that $(c_0\simeq_{\Var^{in}_l} c_0')$,  
if $x\not\in \hat{T}_\ell$, then the values of $x$ are the same at the label $\ell$ in any pair of complete executions $\rho=c_0c_1\cdots c_n$ and $\rho'=c_0'c_1'\cdots c_{n'}'$.
\end{lemma}

\begin{proof}[Proof sketch]
The correctness of the lemma  follows directly from the correctness of
the product construction~\cite{ABBDE16} and Lemmas~\ref{lem:lightweightTA}---\ref{lem:preciseTA}.
\end{proof}

By Lemmas~\ref{lem:lightweightTA}---\ref{lem:product}, we obtain   

\begin{theorem}\label{thm:maintheorem}
Given a safe program $\Prog$,     
if $\hat{T}_\ell=\emptyset$ for any potential side-channel source $(\ell,x)$, $\Prog$ is constant-time secure.
\end{theorem}


\section{Implementation and Evaluation}\label{sec:expe}
We implement our approach in a fully automated prototype tool, named \tool. \tool leverages the static value-flow analysis framework SVF~\cite{sui2016svf} for pre-analysis (i.e., computing call-graph, interprocedural-control-flow graph, points-to information and definition-use chains) and  Phasar~\cite{schubert2019phasar} for building our lightweight taint analysis. In particular, \tool leverages points-to information from SVF to cope with dynamic memory accesses and the recursive implementation of the IFDS framework in Phasar is rewritten as a worklist-based loop implementation for efficiency consideration. \tool utilizes the SMACK toolchain~\cite{rakamaric2014smack} to translate LLVM IR with taint information into Boogie IR~\cite{BarnettCDJL05} and Bam-Bam-Boogieman~\cite{Bam-bam-boogieman} to construct product programs on Boogie IR. The Boogie verifier is used as the underlying safety verification engine. 
Dynamic memory accesses in product programs are handled by SMACK and Boogie.

\begin{table}[t]
\setlength{\tabcolsep}{2pt}
  \centering 
  \caption{Statistics of 87 Benchmarks, where {\bf \#Loc} and 
  {\bf \#Src} show the numbers of lines and potential side-channel sources of the program in Boogie IR, respectively; $\dagger$ indicates that the benchmark is claimed to be constant-time by developers.}
  \vspace{-2mm}
  \scalebox{0.72}{  
      \begin{tabular}{|l|l|r|r| c |l|l|r|r|}
\cline{1-4}\cline{6-9}    {\bf Name}  & {\bf Lib/Algorithm/Function}   & {\bf \#Loc} &  {\bf \#Src} & $\quad$  &  {\bf Name} &  {\bf Lib/Algorithm/Function} &  {\bf \#Loc} &  {\bf \#Src} \\
\cline{1-4}\cline{6-9}    P1\_1 & OpenSSL\_tls1\_cbc\_remove\_padding & 1042  & 71    &   &  {P6\_1} &  {BearSSL\_AES\_small\_decrypt} &  {1395} &  {310} \\
    P1\_2 & OpenSSL\_ssl3\_cbc\_digest\_record & 13089 & 198   &   &  {P6\_2} &  {BearSSL\_AES\_small\_encrypt} &  {945} &  {54} \\
    P1\_3 & OpenSSL\_ssl3\_cbc\_remove\_padding & 351   & 14    &   &  {P6\_3$^\dagger$} &  {BearSSL\_ChaCha20\_ct\_run} &  {2655} &  {316} \\
    P1\_4 & OpenSSL\_ssl3\_cbc\_copy\_mac & 452   & 15    &   &  {P6\_4$^\dagger$} &  {BearSSL\_GHASH\_ctmul32\_br\_ghash\_ctmul} &  {2588} &  {202} \\
\cline{1-4}    p2\_1$^\dagger$ & MAC-then-Encode-then-CBC-Encrypt  & 22036 & 396   &   &  {P6\_5$^\dagger$} &  {BearSSL\_RSA\_i15\_decrypt} &  {5928} &  {229} \\
\cline{1-4}    P3\_1$^\dagger$ & Hacl\_HMAC\_compute\_sha2\_256  & 27404 & 1673  &   &  {P6\_6$^\dagger$} &  {BearSSL\_EC\_p256\_m15\_api\_mul} &  {19535} &  {1895} \\
    P3\_2$^\dagger$ & Hacl\_HMAC\_compute\_blake2b\_32  & 68502 & 5289  &   &  {P6\_7$^\dagger$} &  {BearSSL\_EC\_p256\_m31\_api\_mul} &  {9560} &  {796} \\
    P3\_3$^\dagger$ & Hacl\_HMAC\_legacy\_compute\_sha1  & 2149  & 111   &   &  {P6\_8$^\dagger$} &  {BearSSL\_EC\_p256\_m64\_api\_mul} &  {6450} &  {380} \\
    P3\_4$^\dagger$ & Hacl\_HMAC\_compute\_blake2s\_32  & 56465 & 4309  &   &  {P6\_9$^\dagger$} &  {BearSSL\_GHASH\_ctmul64} &  {1060} &  {10} \\
    P3\_5$^\dagger$ & Hacl\_HMAC\_compute\_sha2\_384  & 34620 & 2183  &   &  {P6\_10$^\dagger$} &  {BearSSL\_AES\_ct\_bitslice\_encrypt} &  {1345} &  {60} \\
    P3\_6$^\dagger$ & Hacl\_HMAC\_compute\_sha2\_512  & 34565 & 2177  &   &  {P6\_11$^\dagger$} &  {BearSSL\_AES\_ct\_bitslice\_decrypt} &  {1788} &  {72} \\
    P3\_7$^\dagger$ & Hacl\_Chacha20\_chacha20\_decrypt  & 3303  & 180   &   &  {P6\_12$^\dagger$} &  {BearSSL\_AES\_ct\_key\_sched} &  {2161} &  {100} \\
    P3\_8$^\dagger$ & Hacl\_Chacha20\_chacha20\_encrypt  & 3285  & 180   &   &  {P6\_13} &  {BearSSL\_AES\_big\_cbc\_key\_schedule} &  {836} &  {290} \\
    P3\_9$^\dagger$ & Hacl\_Curve25519\_64\_ecdh  & 1813  & 70    &   &  {P6\_14} &  {BearSSL\_AES\_big\_cbc\_decrypt} &  {1813} &  {565} \\
    P3\_10$^\dagger$ & Hacl\_Curve25519\_64\_scalarmult  & 1619  & 67    &   &  {P6\_15} &  {BearSSL\_AES\_big\_cbc\_encrypt} &  {1806} &  {565} \\
    P3\_11$^\dagger$ & Hacl\_Curve25519\_64\_secret\_to\_public  & 1701  & 102   &   &  {P6\_16$^\dagger$} &  {BearSSL\_RSA\_i15\_pkcs1\_sign} &  {6464} &  {264} \\
    P3\_12$^\dagger$ & Hacl\_Poly1305\_32\_poly1305\_mac  & 2240  & 121   &   &  {P6\_17} &  {BearSSL\_DES\_table\_cbc\_decrypt} &  {1514} &  {555} \\
    P3\_13$^\dagger$ & Hacl\_Poly1305\_128\_poly1305\_mac  & 72    & 0     &   &  {P6\_18} &  {BearSSL\_DES\_table\_cbc\_encrypt} &  {1501} &  {555} \\
    P3\_14$^\dagger$ & Hacl\_Poly1305\_256\_poly1305\_mac  & 72    & 0     &   &  {P6\_19$^\dagger$} &  {BearSSL\_RSA\_i31\_pkcs1\_sign} &  {5423} &  {235} \\
    P3\_15$^\dagger$ & Hacl\_Curve25519\_51\_ecdh  & 21158 & 2313  &   &  {P6\_20} &  {BearSSL\_Poly1305\_i15\_ChaCha20\_run} &  {3409} &  {347} \\
    P3\_16$^\dagger$ & Hacl\_Curve25519\_51\_scalarmult  & 20964 & 2310  &   &  {P6\_21$^\dagger$} &  {BearSSL\_Poly1305\_ctmul32\_ChaCha20\_run} &  {4802} &  {425} \\
    P3\_17$^\dagger$ & Hacl\_Curve25519\_51\_secret\_to\_public  & 21046 & 2345  &   &  {P6\_22$^\dagger$} &  {BearSSL\_EC\_p256\_m62\_api\_mul} &  {6617} &  {472} \\
\cline{1-4}    P4\_1 & Tongsuo\_curve448\_ossl\_x448 & 6674  & 446   &   &  {P6\_23$^\dagger$} &  {BearSSL\_GHASH\_ctmul\_br\_ghash\_ctmul} &  {1947} &  {125} \\
    P4\_2 & Tongsuo\_curve448\_derive\_pub\_key & 8994  & 575   &   &  {P6\_24$^\dagger$} &  {BearSSL\_AES\_ct64\_bitslice\_encrypt} &  {1348} &  {60} \\
    P4\_3 & Tongsuo\_AES\_decrypt & 3274  & 1382  &   &  {P6\_25$^\dagger$} &  {BearSSL\_AES\_ct64\_bitslice\_decrypt} &  {1791} &  {72} \\
    P4\_4 & Tongsuo\_AES\_encrypt & 2982  & 1126  &   &  {P6\_26$^\dagger$} &  {BearSSL\_AES\_ct64\_key\_sched} &  {2069} &  {119} \\
    P4\_5 & Tongsuo\_constant\_time\_lookup & 339   & 5     &   &  {P6\_27$^\dagger$} &  {BearSSL\_RSA\_i32\_decrypt} &  {4170} &  {179} \\
    P4\_6 & Tongsuo\_curve25519\_derive\_pub\_key & 14744 & 7840  &   &  {P6\_28$^\dagger$} &  {BearSSL\_Poly1305\_ctmul\_ChaCha20\_run} &  {4318} &  {363} \\
    P4\_7 & Tongsuo\_curve25519\_ossl\_x25519 & 3561  & 163   &   &  {P6\_29$^\dagger$} &  {BearSSL\_Poly1305\_ctmulq\_ChaCha20\_run} &  {6761} &  {415} \\
\cline{1-4}    P5\_1 & Mbed TLS\_rsa\_decrypt & 13797 & 788   &   &  {P6\_30$^\dagger$} &  {BearSSL\_DES\_ct\_cbc\_decrypt} &  {1670} &  {49} \\
    P5\_2$^\dagger$ & Mbed TLS\_mpi\_lt\_mpi\_ct  & 411   & 26    &   &  {P6\_31$^\dagger$} &  {BearSSL\_DES\_ct\_cbc\_encrypt} &  {1658} &  {49} \\
\cline{6-9}    P5\_3$^\dagger$ & Mbed TLS\_ct\_rsaes\_pkcs1\_v15\_unpadding  & 629   & 19    &   &  {P8\_1$^\dagger$} &  {FourQlib\_ECC\_double\_eccdouble} &  {28} &  {0} \\
    P5\_4$^\dagger$ & Mbed TLS\_mpi\_safe\_cond\_assign  & 729   & 43    &   &  {P8\_2$^\dagger$} &  {FourQlib\_ECC\_madd\_eccmadd} &  {1632} &  {58} \\
    P5\_5$^\dagger$ & Mbed TLS\_mpi\_core\_lt\_ct  & 214   & 5     &   &  {P8\_3$^\dagger$} &  {FourQlib\_ECC\_norm\_eccnorm} &  {2390} &  {79} \\
\cline{6-9}    P5\_6$^\dagger$ & Mbed TLS\_mpi\_safe\_cond\_swap  & 765   & 46    &   &  {P9\_1} &  {libsodium\_core\_salsa208} &  {1126} &  {9} \\
    P5\_7$^\dagger$ & Mbed TLS\_ct\_memcmp  & 145   & 7     &   &  {P9\_2} &  {libsodium\_aead\_chacha20poly1305\_decrypt} &  {9311} &  {397} \\
    P5\_8$^\dagger$ & Mbed TLS\_ct\_mpi\_uint\_cond\_assign  & 140   & 4     &   &  {P9\_3} &  {libsodium\_core\_salsa20} &  {1094} &  {9} \\
    P5\_9$^\dagger$ & Mbed TLS\_ct\_memcpy\_offset  & 290   & 5     &   &  {P9\_4} &  {libsodium\_stream\_chacha20} &  {6620} &  {259} \\
    P5\_10$^\dagger$ & Mbed TLS\_ct\_base64\_dec\_value  & 309   & 0     &   &  {P9\_5} &  {libsodium\_onetimeauth\_poly1305\_block} &  {778} &  {25} \\
    P5\_11 & Mbed TLS\_DES\_crypt\_cbc & 1991  & 546   &   &  {P9\_6} &  {libsodium\_hash\_sha512\_Transformer} &  {29713} &  {2508} \\
\cline{1-4}    P7\_1$^\dagger$ & libfixedtimefixedpoint\_fix\_pow\_fix\_pow  & 24538 & 80    &   &  {P9\_7} &  {libsodium\_hash\_sha256\_Transformer} &  {23806} &  {2008} \\
\cline{6-9}    P7\_2$^\dagger$ & libfixedtimefixedpoint\_fix\_cmp\_fix\_cmp  & 591   & 0     &   &  {P10\_1} &  {curve25519-donna portable implementation} &  {11737} &  {722} \\
    P7\_3$^\dagger$ & libfixedtimefixedpoint\_fix\_ln\_fix\_ln  & 15200 & 0     &   &  {P10\_2$^\dagger$} &  {curve25519-donna c64 implementation} &  {23936} &  {1599} \\
\cline{6-9}    P7\_4$^\dagger$  & libfixedtimefixedpoint\_fix\_eq\_fix\_eq  & 153   & 0     &  \multicolumn{5}{c}{}  \\
\cline{1-4}    \end{tabular}} %
  \label{tab:statistics}\vspace{-2mm} %
\end{table}%

\smallskip
\noindent
{\bf Research questions}. We investigate the following research questions: 
\begin{enumerate}[label=\textbf{RQ\arabic*.},itemindent=*,leftmargin=*,itemsep=0pt]
\item  How effective and efficient is \tool in proving 
constant-time security?
\item  How efficient is \tool for finding side-channel sources?
\item  What are the respective contributions of the three main steps in \tool?
\end{enumerate}
 
\smallskip
\noindent
{\bf Benchmark}. We collect 87 real-world examples which are implementations from widely used modern cryptographic and SSL/TLS libraries (e.g.,
Tongsuo~\cite{tongsuo23}, BearSSL~\cite{BearSSL23}, Mbed TLS~\cite{MbedTLS23},  HACL*~\cite{hacl*23}, FourQlib~\cite{FourQlib}, libsodium~\cite{libsod23} and OpenSSL~\cite{openSSL23}) and 
(constant-time) implementations from fixed-point arithmetic library (e.g., libfixedtimefixedpoint~\cite{7163051}, curve25519-donna~\cite{donna23}, and MEE-CBC implementations~\cite{almeida2016verifiable}). These benchmarks include cryptographic utilities, arithmetic operations, public and private key cryptography, and algorithms for encryption, decryption, message authentication code (MAC), and digital signature. 
Among 87 examples, 58 are explicitly claimed to be constant-time by developers. 

The statistics of the benchmarks are shown in Table~\ref{tab:statistics}, where benchmarks with constant-time claims are marked by$^\dagger$, \#Loc shows the number of lines of the analyzed Boogie IR code (similar to LLVM bitcode~\cite{ABBDE16}), and \#Src shows the number of all potential side-channel sources. 
\#Loc ranges from 28 to 68,502 (667,518 in total)
and \#Src ranges from 0 to 7,840 (55,060 in total).
Interestingly, we found that 7 (out of 87) benchmarks 
have no potential  side-channel source after being translated
into Boogie IR, indicating that they avoid the use of branching and load/store-related statements.  
We include them because they can be used
to validate the tool implementation and measure
the verification efficiency, and also because some of them were
verified in {\sf ct-verif}~\cite{ABBDE16}.

The experiments were conducted on a machine with Intel Xeon Gold 6342 2.80GHz CPU, 1T RAM, and Ubuntu 20.04.1. 
All the benchmarks were compiled with {\tt clang-12 -c -emit-llvm -O0 -g -Xclang -disable-O0-optnone} and  {\tt opt -mem2reg},
the same as {\sf ct-verif}~\cite{ABBDE16}. In particular,
{\tt O0} disables optimizations to avoid compiler-introduced leakage, {\tt disable-O0-optnone} disables the `optnone’ pass which affects `-mem2reg’ and {\tt mem2reg} is used to reduce redundant load/store instructions. 

\begin{table}[t]
\setlength{\tabcolsep}{2.5pt}
  \centering
  \caption{Results of verifying constant-time implementations, where TO denotes time out (1 hour).}
  \vspace{-2mm}
\scalebox{0.72}{  
    \begin{tabular}{|l|r|r|r|l|r|r|r|l|r|r|r|l|r|r|}
\cline{1-3}\cline{5-7}\cline{9-11}\cline{13-15}    {\bf Name}  &  {\tool} &  {\sf ct-verif} &   $\quad$    & {\bf Name}  &  {\tool} &  {\sf ct-verif} &    $\quad$     & {\bf Name}  &  {\tool} &  {\sf ct-verif} &    $\quad$     & {\bf Name} &  {\tool} &  {\sf ct-verif} \\
\cline{1-3}\cline{5-7}\cline{9-11}\cline{13-15}    P1\_1 & 7.55  & 7.53  &       & P3\_14 $^\dagger$ & 0.04  & 6.93  &       & P6\_3 $^\dagger$ & 0.10  & 9.25  &       & P6\_30 $^\dagger$ & 0.07  & 8.26  \\
    P1\_2 & 1.12  & 26.74  &       & P3\_15 $^\dagger$ & 0.70  & {\cellcolor{red!20}TO} &       & P6\_4 $^\dagger$ & 0.20  & 9.04  &       & P6\_31 $^\dagger$ & 0.07  & 8.51  \\
    P1\_3 & 0.04  & 7.25  &       & P3\_16 $^\dagger$ & 0.67  & {\cellcolor{red!20}TO} &       & P6\_6 $^\dagger$ & 23.24  & {\cellcolor{red!20}TO} &       & P7\_1 $^\dagger$ & 0.64  & \cellcolor{blue!20}TO  \\
    P1\_4 & 0.05  & 7.89  &       & P3\_17 $^\dagger$ & 0.63  & {\cellcolor{red!20}TO} &       & P6\_7 $^\dagger$ & 0.45  & {\cellcolor{red!20}TO} &       & P7\_2 $^\dagger$ & 0.51  & 7.42  \\
    p2\_1 $^\dagger$ & 1.44  & 134.72  &       & P4\_1 & 0.53  & {\cellcolor{red!20}TO} &       & P6\_8 $^\dagger$ & 0.21  & 794.79  &       & P7\_3 $^\dagger$ & 0.59  & 21.36  \\
    P3\_1 $^\dagger$ & 108.80  & {\cellcolor{red!20}TO} &       & P4\_2 & 0.57  & {\cellcolor{red!20}TO} &       & P6\_9 $^\dagger$ & 0.05  & 7.65  &       & P7\_4 $^\dagger$ & 0.52  & 7.37  \\
    P3\_2 $^\dagger$ & 177.48  & {\cellcolor{red!20}TO} &       & P4\_5 & 0.04  & 7.42  &       & P6\_10 $^\dagger$ & 0.07  & 7.79  &       & P8\_1 $^\dagger$ & 0.20  & 6.87  \\
    P3\_3 $^\dagger$ & 104.40  & {\cellcolor{red!20}TO} &       & P4\_6 & 0.35  & 177.25  &       & P6\_11 $^\dagger$ & 0.08  & 8.22  &       & P8\_2 $^\dagger$ & 0.40  & 21.69  \\
    P3\_4 $^\dagger$ & 169.71  & {\cellcolor{red!20}TO} &       & P4\_7 & 0.34  & 14.16  &       & P6\_12 $^\dagger$ & 0.07  & 8.66  &       & P8\_3 $^\dagger$ & 0.39  & 37.73  \\
    P3\_5 $^\dagger$ & 109.19  & {\cellcolor{red!20}TO} &       & P5\_2 $^\dagger$ & 0.08  & 7.27  &       & P6\_20 & 0.11  & 12.38  &       & P9\_1 & 0.14  & 7.86  \\
    P3\_6 $^\dagger$ & 122.85  & {\cellcolor{red!20}TO} &       & P5\_3 $^\dagger$ & 0.08  & 7.36  &       & P6\_21 $^\dagger$ & 0.14  & 14.49  &       & P9\_2 & 25.69  & 23.02  \\
    P3\_7 $^\dagger$ & 0.12  & 30.51  &       & P5\_4 $^\dagger$ & 1.78  & 7.49  &       & P6\_22 $^\dagger$ & 0.24  & 977.27  &       & P9\_3 & 0.13  & 7.97  \\
    P3\_8 $^\dagger$ & 0.13  & 30.65  &       & P5\_5 $^\dagger$ & 0.07  & 7.05  &       & P6\_23 $^\dagger$ & 0.11  & 8.40  &       & P9\_4 & 16.05  & \cellcolor{blue!20} 11.91   \\
    P3\_9 $^\dagger$ & 0.10  & 8.94  &       & P5\_6 $^\dagger$ & 1.86  & 7.57  &       & P6\_24 $^\dagger$ & 0.06  & 7.93  &       & P9\_5 & 0.06  & 7.25  \\
    P3\_10 $^\dagger$ & 0.09  & 8.91  &       & P5\_7 $^\dagger$ & 0.08  & 6.91  &       & P6\_25 $^\dagger$ & 0.08  & 8.41  &       & P9\_6 & 4.25  & 44.12  \\
    P3\_11 $^\dagger$ & 0.09  & 9.26  &       & P5\_8 $^\dagger$ & 0.08  & 6.46  &       & P6\_26 $^\dagger$ & 0.06  & 8.94  &       & P9\_7 & 3.34  & 32.98  \\
    P3\_12 $^\dagger$ & 0.09  & 9.28  &       & P5\_9 $^\dagger$ & 0.07  & 6.63  &       & P6\_28 $^\dagger$ & 0.12  & 12.23  &       & P10\_1 & 0.55  & {\cellcolor{red!20}TO} \\
    P3\_13 $^\dagger$ & 0.04  & 7.05  &       & P5\_10 $^\dagger$ & 0.08  & 7.00  &       & P6\_29 $^\dagger$ & 0.17  & 13.86  &       & P10\_2 $^\dagger$ & 10.89  & 33.18  \\
\cline{1-3}\cline{5-7}\cline{9-11}\cline{13-15}    \end{tabular}}%
  \label{tab:2}\vspace{-2mm}%
\end{table}%
\subsection{Verifying Constant-time Implementations (RQ1, RQ3)} \label{sect:rq1}
To answer RQ1, we verify all 87 programs. \tool returns "\emph{proved}" for 72 benchmarks. (One will see later that the other 15 programs are indeed not constant-time and 
the constant-time claim made by the developer may be incorrect.) We mainly compare \tool with {\sf ct-verif} which is the only available tool at the LLVM IR level. 

The results are reported in Table~\ref{tab:2}, where 
verification time is given in seconds and TO denotes time out (1 hour).
{\sf ct-verif} fails to prove 
16 constant-time programs: it runs out of time on 14 programs without outputting any potential side-channel sources (highlighted in \textcolor{red!60}{red} color) and outputs inconclusive verification results on 2 programs within 1 hour (highlighted in \textcolor{blue!60}{blue} color). The latter is because of the loop in P7\_1 (cf. Example 1) and the type-casting of pointers in P9\_4. This indicates that  
the theoretical completeness of self-composition based approaches 
may be compromised by the limitation of safety verification. 
On the other 56 constant-time programs that can be proved
by {\sf ct-verif} and \tool, \tool is about 50 times faster.

To partially answer RQ3, we inspect at which step a program can be proved
in the verification process of \tool. We found that 
all these programs (except for P1\_1 and P9\_4) can be proved 
using the lightweight taint analysis solely. 
The lightweight taint analysis cannot determine 4 (resp. 2) potential side-channel sources of P1\_1 (resp. P9\_4) due to an infeasible branch condition
(resp. index-insensitive), which were resolved by the precise taint analysis. While the individual cost of the two taint analyses is lower than {\sf ct-verif}, the accumulated cost is slightly 
higher. 

In summary, \tool can return conclusive results for these benchmarks. Most constant-time implementations can be proved by the lightweight taint analysis solely, whereas
the remaining ones can be proved by the precise taint analysis.
\tool significantly outperforms the state-of-the-art 
tool {\sf ct-verif} for real-world constant-time implementations.


\subsection{Verifying Non-constant-time Implementations (RQ2, RQ3)}
To answer RQ2, we check the remaining 15 cases which are deemed to be vulnerable. The results are shown in Table~\ref{tab:3}, where \toollight refers to \tool \emph{without} the precise taint analysis (i.e., the 2nd main step), 
a:b in Time (s) respectively denote the execution time after the 2nd and 3rd main step (note: the 1st step is very efficient 
whose time is negligible), TO denotes time out (6 hours), x:y:z (resp. x:z) in \#Src respectively denote the numbers of side-channel sources reported after the 1st, 2nd and 3rd (resp. 1st and 3rd) main step, and $^{\star}$ indicates that some invalid loop invariants are added by {\sf ct-verif} and thus may miss side-channel sources.
We have manually checked all these final potential side-channel sources 
with their corresponding execution traces produced by \tool,
and found that all of them are genuine timing side-channel vulnerabilities (i.e., not false positives). In particular, P6\_5, P6\_16, P6\_19 and P6\_27 (from BearSSL) were claimed to be of constant-time.\footnote{
We have reported all the potential vulnerabilities in Table~\ref{tab:3} to the respective developer(s), 
and received confirmations for Mbed TLS and BearSSL.}

\begin{table}[t]
  \centering\setlength{\tabcolsep}{2pt}
  \caption{Results of verifying non-constant-time implementations, where \toollight  refers to \tool \emph{without} the precise taint analysis (i.e., the 2nd main step), 
  a:b in Time (s) respectively denote the execution time 
  after the 2nd and 3rd main step (note: the 1st step is very efficient 
  whose time is negligible), TO denotes time out (6 hours), x:y:z (resp. x:z) in \#Src respectively denote the numbers of side-channel sources reported after the 1st, 2nd and 3rd (resp. 1st and 3rd) main step, and $^{\star}$ indicates that
  some invalid loop invariants are added by {\sf ct-verif} and thus
  may miss side-channel sources.}
   \vspace{-2mm}
  \scalebox{0.72}{
   \begin{tabular}{|l|r|r|r|r|r|r| c  |l|r|r|r|r|r|r|}
\cline{1-7}\cline{9-15}    
\multirow{2}{*}{\bf Name} & \multicolumn{2}{c|}{\tool} & \multicolumn{2}{c|}{\sf ct-verif} & \multicolumn{2}{c|}{\toollight} & \multicolumn{1}{c|}{$\quad$}
& \multirow{2}{*}{\bf Name} & \multicolumn{2}{c|}{\tool} & \multicolumn{2}{c|}{\sf ct-verif} & \multicolumn{2}{c|}{\toollight} \\ \cline{2-7}\cline{10-15}        
& Time (s) & \#Src & Time (s) & \#Src & Time (s) & \#Src &   &   &  {Time (s)} &  {\#Src} &  {Time (s)} &  {\#Src} &  {Time (s)} &  {\#Src} \\\cline{1-7}\cline{9-15}  

P4\_3 & 24.05:50.84  & 49:48:48 & 17.69  & 32$^{\star}$    & 24.77  & 49:48 &   & \multicolumn{1}{l|} {P6\_14} & {13.06:26.57}  & {32:32:32} & {10.58}  & {0$^{\star}$}    & {12.64}  & {32:32} \\
    P4\_4 & 21.03:44.76  & 49:48:48 & 16.73  & 32$^{\star}$    & 21.96  & 49:48 &   & \multicolumn{1}{l|}  {P6\_15} & {13.06:26.35}  & {32:32:32} & {10.37}  & {0$^{\star}$}    & {12.42}  & {32:32}  \\
    P5\_1 & \cellcolor{blue!20}TO:TO   & 48:--:-- & \cellcolor{red!20}TO    & -- & \cellcolor{blue!20}TO    & 48:-- &   & \multicolumn{1}{l|}  {P6\_16$^\dagger$} &  {\cellcolor{blue!20}2561.81:TO} &  {107:71:53} &  {\cellcolor{blue!20}TO} &  {21$^{\star}$} &  {\cellcolor{blue!20}TO} &  {107:55}\\
    P5\_11 & 14.10:30.68  & 26:16:16 & 11.81  & 16    & 15.20  & 26:16 &   & \multicolumn{1}{l|} {P6\_17} &  {9.75:20.31} &  {11:8:8} &  {9.29} &  {8} &  {9.48} &  {11:8} \\
    P6\_1 & 9.36:19.27  & 12:1:1 & 8.46  & 1     & 9.13  & 12:1  &   & \multicolumn{1}{l|} {P6\_18} &  {9.77:20.12} &  {13:8:8} &  {9.26} &  {8} &  {9.63} &  {13:8} \\
    P6\_2 & 8.48:17.24  & 12:1:1 & 7.45  & 1     & 8.46  & 12:1  &   & \multicolumn{1}{l|} {P6\_19$^\dagger$} &  {\cellcolor{blue!20}3515.77:TO} &  {142:70:39} &  {\cellcolor{blue!20}TO} &  {21$^{\star}$} &  {\cellcolor{blue!20}TO} &  {142:46} \\
    P6\_5$^\dagger$ & 983.86:{\cellcolor{blue!20}TO}  & 107:71:50 & \cellcolor{blue!20}TO    & 21$^{\star}$    & \cellcolor{blue!20}TO    & 107:49 &   & \multicolumn{1}{l|} {P6\_27$^\dagger$} &  {131.34:503.23} &  {177:45:45} &  {450.57} &  {22$^{\star}$} &  {580.50} &  {177:45} \\
\cline{9-15}  P6\_13 & 8.85:17.87  & 6:4:4 & 8.45  & 4     & 8.51  & 6:4    &    \multicolumn{6}{c}{}    \\
\cline{1-7}   \end{tabular}}%
  \label{tab:3}\vspace{-3mm}%
\end{table}%

Both \tool and {\sf ct-verif} ran out of time on 4 programs (P5\_1, P6\_5, P6\_16 and P6\_19). 
We note that on P5\_1 our lightweight taint analysis is still able to find 48 potential side-channel sources based on which we manually confirm that P5\_1 is not constant-time.
\tool finally found more side-channel sources than {\sf ct-verif} on 
9 (out of 15) programs (marked by $^\star$ in Table~\ref{tab:3}), because {\sf ct-verif} misses side-channel sources when some additional loop invariants cannot be proved, while \tool still works after automatically removing such loop invariants.
On the other programs, \tool and {\sf ct-verif} report the same side-channel sources. 
In terms of efficiency, \tool is slightly slower than
{\sf ct-verif}. It is not surprising as \tool involves
three main steps while the last two main steps have to perform safety checking. We remark that
{\sf ct-verif} may miss potential
side-channel sources since it skips part of the program (e.g., the loop body) when additional loop invariants cannot be proved, which 
explains its reduced execution time.

To partially answer RQ3, we analyze the respective number of (potential) side-channel sources reported by the three main steps of \tool and two main steps of \toollight (in the form of x:y:z and x:z in Table~\ref{tab:3}).
We find that the lightweight taint analysis can determine a large number of 
potential side-channel sources, the precise taint analysis can resolve
the remaining few unsolved ones (i.e., P4\_3, P4\_4, 
P5\_11, P6\_1, P6\_2, P6\_5, P6\_13, P6\_16, P6\_17, P6\_18, P6\_19, and P6\_27), and
the left-over ones are often vulnerabilities (i.e., P4\_3, P4\_4, 
P5\_11, P6\_1, P6\_2, P6\_13, P6\_17, P6\_18 and P6\_27). 
By comparing with \toollight, we can observe that
disabling the precise taint analysis in \tool (i.e., the 2nd main step) may both improve (i.e., P6\_5) and degrade (i.e., P6\_16 and P6\_19) the capability of finding side-channel sources while reducing the verification time. 
Nevertheless, the precise taint analyses are still useful for finding all the potential side-channel sources when the 3rd main step runs out of time (i.e., P6\_5, P6\_16, and P6\_19), and the first two main steps (i.e., the lightweight and precise taint analysis) are often able to find all the side-channel sources (i.e., P4\_3, P4\_4,
P5\_11, P6\_1, P6\_2, P6\_13, P6\_14, P6\_15, P6\_17, P6\_18 and P6\_27) with comparable or less execution time than \toollight.

\begin{figure}[H]
  \centering\vspace{-1mm}
\begin{lstlisting}
uint32_t br_rsa_i15_private(const br_rsa_private_key *sk){
    const unsigned char *p = sk->p;   size_t plen = sk->plen;
	while (plen > 0 && *p == 0) { p++; plen--;} }
  \end{lstlisting}\vspace{-3mm}
  \caption{Simplified fragment of P6\_5 taken from the BearSSL library.}
  \label{fig:part6-5}\vspace{-3mm}
\end{figure}

\smallskip
\noindent
{\bf Case study}.
Fig.~\ref{fig:part6-5} shows one side-channel source $(4, *p ==0)$ of P6\_5 from BearSSL. Variable {\tt p} points to a buffer storing a large prime (for RSA) which is secret but the loop condition {\tt *p == 0} depends upon the content of the buffer. It leaks the number of leading zero of the large prime.

In summary, \tool is efficient and effective for finding side-channel sources 
and significantly outperforms {\sf ct-verif}. Both taint analyses can determine potential side-channel sources, reducing the cost of the subsequent safety verification and manual validation. 

\smallskip
\noindent
{\bf Further comparison}.
We also compare \tool with two sound but incomplete tools
Verasco~\cite{BPT17} and BINSEC (the latest version of Binsec/Rel~\cite{DanielBR23}).
Verasco performs taint analysis by abstract interpretation with bounded loops, while BINSEC uses relational
symbolic execution with bounded paths. We 
use relatively large benchmarks provided by the respective tools with security annotations to reduce the engineering efforts of modifying benchmarks, 
and compile benchmarks using the constant-time preserving compiler~\cite{BartheBGHLPT20}.

We observe that \tool is significantly more efficient than Verasco on
almost all the benchmarks (80.63 seconds vs. 1,127.63 seconds in total). Moreover, Verasco cannot output traces and values of scalar input variables leading to potential violations, and misses some timing side-channel sources on the DES implementation.
Compared with BINSEC, \tool is about 2 times faster than BINSEC when BINSEC’s execution time is independent of input size (e.g., curve25519-donna), and the speedup becomes more significant when BINSEC’s execution time increases with input size, e.g., BINSEC takes 0.36, 10.28 and 102.48 seconds on the libsodium\_chacha20\_xor implementation when the input size is 256, 256$\times$512 and 256$\times$5120 bits while
\tool takes 11.21 seconds without limiting the input size.

\subsection{Threats to Validity}
{\bf Internal threats.} The major internal threat to our evaluation is the correctness of the implementation of \tool which relies on various open-source frameworks and tools. To mitigate this threat, we compared  the results of \tool and {\sf ct-verif}, and manually analyzed all the side-channel sources discovered by both tools to confirm
 the correctness of \tool. Moreover, we note that these open-source frameworks and tools have been widely used for years so we would have reasonable confidence in their quality. A more scientific way is to build a verified toolchain, which however requires more resources and thus is left as interesting future work.

\smallskip
\noindent {\bf External threats.}
The major external threat to our evaluation is the benchmarks, as the performance of \tool may vary with benchmarks. To mitigate this threat, we consider both constant-time and non-constant-time implementations from 
widely used modern cryptographic and SSL/TLS libraries, as well as 
benchmarks used in prior work, e.g.,~\cite{ABBDE16,BPT17,DanielBR23}. Furthermore, the benchmarks are selected to be diverse in terms of both types and sizes. Another external threat is the translation from source code to LLVM IR which may introduce violations of constant-time, as demonstrated by~\cite{DanielBR23}. To mitigate this threat, benchmarks are only optimized by {\tt -mem2reg} which reduces redundant address-taken variables in LLVM IR. 


\section{Related Work} \label{sect:related}
Timing side-channels have received considerable attention (cf.~\cite{GeimerVRDBM23} for a survey). We roughly classify the current approaches into three categories. The first class does not rely on verification; the second class mostly leverages program analysis; the third class is largely based on verification.
There are verification approaches for other side-channels (e.g., ~\cite{QinJSCX22,ZGSW18,GXZSC19,FanSCZL22,BBDFGS15,GaoXSC2020,GZSC19,GXSZSC2020,GaoZSCS23})
and time-balance (e.g.,~\cite{AthanasiouCEMST18,AGHK17,BrennanSB18}), which are orthogonal to this work.

\smallskip
\noindent \textbf{Approaches based on concrete execution.} A large number of approaches have been proposed for detecting and/or quantifying timing side-channel leakages in terms of e.g., channel capacity and Shannon entropy.
To this end, for instance, 
DATA~\cite{WeiserZSMMS18}, 
ct-fuzz~\cite{HeEC20}  
and CacheQL~\cite{YLW23} 
make use of concrete execution.
CANAL~\cite{SungPW18},
Abacus~\cite{BaoWLLW21}, 
and ENCIDER~\cite{YavuzFHBBT23} leverage dynamic symbolic execution, symbolic execution on individual concrete execution traces, and concolic execution. 
CaType~\cite{JiangBWL022} detects timing side-channel vulnerabilities by applying type inference on individual concrete execution traces. 
In general, this class of approaches are often effective in bug-finding,  
but cannot prove the absence of timing side-channel leakages.

\smallskip
\noindent\textbf{Approaches based on program analysis.}
This class of work is more formal 
which often is able to prove the absence of timing side-channel leakage.
CacheAudit~\cite{DoychevFKMR13} bounds
timing side-channel leakages by over-approximating side-channel observations using abstract interpretation.
CacheS~\cite{WBLWZWu19} applies abstract interpretation, but its implementation is unsound due to its  
imprecise treatment of memory.
Taint analysis and security type systems have also been applied to detect side-channel leakage by tracking information flow
of the secrets~\cite{BBCLP14,BPT17,WattRPCS19}, varying in accuracy and efficiency. 
Moreover, the work \cite{BangAPPB16,Bultan19,SahaGLBB23} bounds the information leakage via symbolic execution and model-counting.
While these approaches are often sound, they may raise false positives. Program transformations have been proposed to eliminate potential side-channel sources~\cite{Agat00,WuGS018,CauligiSJBWRGBJ19}. 
Precise detection approaches can reduce the number of
potential side-channel sources but the adopted program transformations may bloat the program, making them less efficient to run.

\smallskip
\noindent \textbf{Approaches based on verification.}
There are mainly two approaches in this class, i.e.,  
self-composition~\cite{AlmeidaBPV13}
and relational symbolic execution~\cite{FarinaCG19}.
ct-verif~\cite{ABBDE16} uses a variant of self-composition (i.e., cross-product) to improve the efficiency of safety verification. ct-verif is complete assuming the completeness of the underlying safety checker.
Binsec/Rel~\cite{DanielBR20,DanielBR23} uses relational symbolic execution enhanced with secret-dependency tracking,
untainting and fault-packing to improve the efficiency. Due to the path explosion problem and unsupported dynamic memory allocation, it is only complete up to a given depth of paths, and the size of the symbolic input (keys, plaintext) has to be fixed. 
ct-verif targets LLVM IR, a target-independent low-level language, 
while Binsec/Rel targets
binary executables. (Note that the compilation from LLVM IR to binary executable may introduce vulnerabilities.) 
 
Our work generally falls into the third category.  
Similar to ct-verif, we focus on LLVM IR instead of binary executable or source code.  
Compared with the existing verification approaches, 
we extensively leverage taint analysis
to facilitate self-composition, especially for   
reducing safety checks and simplifying the self-composed program. A similar idea was adopted in Binsec/Rel, but in a different way. Our new methodology significantly improves both efficiency and effectiveness.

\smallskip
\noindent\textbf{Consideration of micro-architecture.} The above approaches  
did not address micro-architectural features which have also been studied in literature. 
For instance, Constantine~\cite{BorrelloDQG21} uses dynamic taint analysis;
Oo7~\cite{WangCGMR21} uses static taint analysis;
Binsec/Haunted~\cite{DanielBR21} uses relational symbolic execution; Pitchfork~\cite{CauligiDGTSRB20}, 
Spectector~\cite{GuarnieriKMRS20}, SpecuSym~\cite{GuoCLCWW020} 
and KleeSpectre~\cite{WangCBMR20} leverage (dynamic) symbolic execution; Blade~\cite{VassenaDGCKJTS21} uses type system; \cite{Wu019} uses abstract interpretation. 

Our work is orthogonal to these investigations, as they usually tackle the vulnerabilities brought by micro-architectural features with the assumption that 
the program itself is of constant-time.


%

\section{Conclusion} \label{sec:concl}
 
In this paper, we have provided practical verification approaches for constant-time implementations of cryptographic libraries. Our methods are based on a novel synergy of taint analysis and safety verification of self-composed programs. We have implemented a cross-platform and fully automated tool {\tool} working on LLVM IR. The tool has been extensively evaluated on a large set of real-world benchmarks from modern cryptographic and SSL/TLS and fixed-point arithmetic libraries.  
The experimental results have confirmed the efficacy of our approaches. In particular, compared to the state-of-the-art tool {\sf ct-verif}, \tool typically demonstrates 2-3 orders of magnitude of improvement, proving more programs and finding new timing leaks. 

At the methodology level, we showcase that a combination of lightweight approaches (taint analysis) and heavyweight approaches (self-composition and safety checking) can yield efficiency and completeness. In particular, we demonstrate that self-composition based approaches, which are normally considered to be powerful but costly, can be made scalable with the aid of static analysis. 

\section*{Data Availability}
To foster further research, source code and benchmarks are available at~\cite{CTProver}. 

\section*{Acknowledgement}
This work is supported by the Amazon Research Award from Amazon Web Services, National Natural Science Foundation of China under
grants No. 62072309 and 61872340,
CAS Project for Young Scientists in Basic Research (YSBR-040),  
ISCAS New Cultivation Project (ISCAS-PYFX-202201), 
ISCAS Fundamental Research Project (ISCAS-JCZD-202302),
Overseas grants from the State Key Laboratory of Novel Software Technology, Nanjing University (KFKT2022A03, KFKT2023A04), 
and Birkbeck BEI School Project (EFFECT).
	
	


\begin{thebibliography}{84}
	
	
	\ifx \showCODEN    \undefined \def \showCODEN     #1{\unskip}     \fi
	\ifx \showDOI      \undefined \def \showDOI       #1{#1}\fi
	\ifx \showISBNx    \undefined \def \showISBNx     #1{\unskip}     \fi
	\ifx \showISBNxiii \undefined \def \showISBNxiii  #1{\unskip}     \fi
	\ifx \showISSN     \undefined \def \showISSN      #1{\unskip}     \fi
	\ifx \showLCCN     \undefined \def \showLCCN      #1{\unskip}     \fi
	\ifx \shownote     \undefined \def \shownote      #1{#1}          \fi
	\ifx \showarticletitle \undefined \def \showarticletitle #1{#1}   \fi
	\ifx \showURL      \undefined \def \showURL       {\relax}        \fi
	\providecommand\bibfield[2]{#2}
	\providecommand\bibinfo[2]{#2}
	\providecommand\natexlab[1]{#1}
	\providecommand\showeprint[2][]{arXiv:#2}
	
	\bibitem[don(2023)]%
	{donna23}
	\bibinfo{year}{2023}\natexlab{}.
	\newblock \bibinfo{booktitle}{\emph{{curve25519-donna}}}.
	\newblock
	\newblock
	\shownote{\url{https://github.com/agl/curve25519-donna}}.
	
	
	\bibitem[hac(2023)]%
	{hacl*23}
	\bibinfo{year}{2023}\natexlab{}.
	\newblock \bibinfo{booktitle}{\emph{{HACL*}}}.
	\newblock
	\newblock
	\shownote{\url{https://github.com/hacl-star/hacl-star}}.
	
	
	\bibitem[lib(2023)]%
	{libsod23}
	\bibinfo{year}{2023}\natexlab{}.
	\newblock \bibinfo{booktitle}{\emph{{libsodium}}}.
	\newblock
	\newblock
	\shownote{\url{https://doc.libsodium.org/}}.
	
	
	\bibitem[ope(2023)]%
	{openSSL23}
	\bibinfo{year}{2023}\natexlab{}.
	\newblock \bibinfo{booktitle}{\emph{{OpenSSL}}}.
	\newblock
	\newblock
	\shownote{\url{https://www.openssl.org}}.
	
	
	\bibitem[ton(2023)]%
	{tongsuo23}
	\bibinfo{year}{2023}\natexlab{}.
	\newblock \bibinfo{booktitle}{\emph{{Tongsuo}}}.
	\newblock
	\newblock
	\shownote{\url{https://github.com/Tongsuo-Project/Tongsuo}}.
	
	
	\bibitem[CTP(2024)]%
	{CTProver}
	\bibinfo{year}{2024}\natexlab{}.
	\newblock \bibinfo{booktitle}{\emph{{CT-Prover}}}.
	\newblock
	\urldef\tempurl%
	\url{https://doi.org/10.5281/zenodo.10683405}
	\showDOI{\tempurl}
	
	
	\bibitem[Agat(2000)]%
	{Agat00}
	\bibfield{author}{\bibinfo{person}{Johan Agat}.}
	\bibinfo{year}{2000}\natexlab{}.
	\newblock \showarticletitle{Transforming Out Timing Leaks}. In
	\bibinfo{booktitle}{\emph{Proc. of the 27th {ACM} {SIGPLAN-SIGACT} Symposium
			on Principles of Programming Languages ({POPL})}}. \bibinfo{pages}{40--53}.
	\newblock
	
	
	\bibitem[Albrecht and Paterson(2016)]%
	{AlbrechtP16}
	\bibfield{author}{\bibinfo{person}{Martin~R. Albrecht} {and}
		\bibinfo{person}{Kenneth~G. Paterson}.} \bibinfo{year}{2016}\natexlab{}.
	\newblock \showarticletitle{Lucky Microseconds: {A} Timing Attack on Amazon's
		s2n Implementation of {TLS}}. In \bibinfo{booktitle}{\emph{Proc. of the 35th
			Annual International Conference on the Theory and Applications of
			Cryptographic Techniques}}. \bibinfo{pages}{622--643}.
	\newblock
	\urldef\tempurl%
	\url{https://doi.org/10.1007/978-3-662-49890-3\_24}
	\showDOI{\tempurl}
	
	
	\bibitem[AlFardan and Paterson(2013)]%
	{AlFardanP13}
	\bibfield{author}{\bibinfo{person}{Nadhem~J. AlFardan} {and}
		\bibinfo{person}{Kenneth~G. Paterson}.} \bibinfo{year}{2013}\natexlab{}.
	\newblock \showarticletitle{Lucky Thirteen: Breaking the {TLS} and {DTLS}
		Record Protocols}. In \bibinfo{booktitle}{\emph{Proc. of the 2013 {IEEE}
			Symposium on Security and Privacy}}. \bibinfo{pages}{526--540}.
	\newblock
	\urldef\tempurl%
	\url{https://doi.org/10.1109/SP.2013.42}
	\showDOI{\tempurl}
	
	
	\bibitem[Almeida et~al\mbox{.}(2016a)]%
	{almeida2016verifiable}
	\bibfield{author}{\bibinfo{person}{Jos{\'e}~Bacelar Almeida},
		\bibinfo{person}{Manuel Barbosa}, \bibinfo{person}{Gilles Barthe}, {and}
		\bibinfo{person}{Fran{\c{c}}ois Dupressoir}.}
	\bibinfo{year}{2016}\natexlab{a}.
	\newblock \showarticletitle{Verifiable side-channel security of cryptographic
		implementations: constant-time MEE-CBC}. In \bibinfo{booktitle}{\emph{Proc.
			of the 23rd International Conference on Fast Software Encryption}}.
	\bibinfo{pages}{163--184}.
	\newblock
	
	
	\bibitem[Almeida et~al\mbox{.}(2016b)]%
	{ABBDE16}
	\bibfield{author}{\bibinfo{person}{Jos{\'{e}}~Bacelar Almeida},
		\bibinfo{person}{Manuel Barbosa}, \bibinfo{person}{Gilles Barthe},
		\bibinfo{person}{Fran{\c{c}}ois Dupressoir}, {and} \bibinfo{person}{Michael
			Emmi}.} \bibinfo{year}{2016}\natexlab{b}.
	\newblock \showarticletitle{Verifying Constant-Time Implementations}. In
	\bibinfo{booktitle}{\emph{Proc. of the 25th {USENIX} Security Symposium}}.
	\bibinfo{pages}{53--70}.
	\newblock
	
	
	\bibitem[Almeida et~al\mbox{.}(2013)]%
	{AlmeidaBPV13}
	\bibfield{author}{\bibinfo{person}{Jos{\'{e}}~Bacelar Almeida},
		\bibinfo{person}{Manuel Barbosa}, \bibinfo{person}{Jorge~Sousa Pinto}, {and}
		\bibinfo{person}{B{\'{a}}rbara Vieira}.} \bibinfo{year}{2013}\natexlab{}.
	\newblock \showarticletitle{Formal verification of side-channel countermeasures
		using self-composition}.
	\newblock \bibinfo{journal}{\emph{Sci. Comput. Program.}} \bibinfo{volume}{78},
	\bibinfo{number}{7} (\bibinfo{year}{2013}), \bibinfo{pages}{796--812}.
	\newblock
	\urldef\tempurl%
	\url{https://doi.org/10.1016/J.SCICO.2011.10.008}
	\showDOI{\tempurl}
	
	
	\bibitem[Andrysco et~al\mbox{.}(2015)]%
	{7163051}
	\bibfield{author}{\bibinfo{person}{Marc Andrysco}, \bibinfo{person}{David
			Kohlbrenner}, \bibinfo{person}{Keaton Mowery}, \bibinfo{person}{Ranjit
			Jhala}, \bibinfo{person}{Sorin Lerner}, {and} \bibinfo{person}{Hovav
			Shacham}.} \bibinfo{year}{2015}\natexlab{}.
	\newblock \showarticletitle{On Subnormal Floating Point and Abnormal Timing}.
	In \bibinfo{booktitle}{\emph{Proc. of the IEEE Symposium on Security and
			Privacy}}. \bibinfo{pages}{623--639}.
	\newblock
	\urldef\tempurl%
	\url{https://doi.org/10.1109/SP.2015.44}
	\showDOI{\tempurl}
	
	
	\bibitem[Antonopoulos et~al\mbox{.}(2017)]%
	{AGHK17}
	\bibfield{author}{\bibinfo{person}{Timos Antonopoulos}, \bibinfo{person}{Paul
			Gazzillo}, \bibinfo{person}{Michael Hicks}, \bibinfo{person}{Eric Koskinen},
		\bibinfo{person}{Tachio Terauchi}, {and} \bibinfo{person}{Shiyi Wei}.}
	\bibinfo{year}{2017}\natexlab{}.
	\newblock \showarticletitle{Decomposition instead of self-composition for
		proving the absence of timing channels}. In \bibinfo{booktitle}{\emph{Proc.
			of the 38th {ACM} {SIGPLAN} Conference on Programming Language Design and
			Implementation}}. \bibinfo{pages}{362--375}.
	\newblock
	\urldef\tempurl%
	\url{https://doi.org/10.1145/3062341.3062378}
	\showDOI{\tempurl}
	
	
	\bibitem[Athanasiou et~al\mbox{.}(2018)]%
	{AthanasiouCEMST18}
	\bibfield{author}{\bibinfo{person}{Konstantinos Athanasiou},
		\bibinfo{person}{Byron Cook}, \bibinfo{person}{Michael Emmi},
		\bibinfo{person}{Colm MacC{\'{a}}rthaigh}, \bibinfo{person}{Daniel
			Schwartz{-}Narbonne}, {and} \bibinfo{person}{Serdar Tasiran}.}
	\bibinfo{year}{2018}\natexlab{}.
	\newblock \showarticletitle{SideTrail: Verifying Time-Balancing of
		Cryptosystems}. In \bibinfo{booktitle}{\emph{Proc. of the 10th International
			Conference on Verified Software. Theories, Tools, and Experiments}}.
	\bibinfo{pages}{215--228}.
	\newblock
	
	
	\bibitem[{AWS}(2023)]%
	{s2n23}
	\bibfield{author}{\bibinfo{person}{{AWS}}.} \bibinfo{year}{2023}\natexlab{}.
	\newblock \bibinfo{booktitle}{\emph{{s2n-tls}}}.
	\newblock
	\newblock
	\shownote{\url{ https://github.com/aws/s2n-tls}}.
	
	
	\bibitem[Bang et~al\mbox{.}(2016)]%
	{BangAPPB16}
	\bibfield{author}{\bibinfo{person}{Lucas Bang}, \bibinfo{person}{Abdulbaki
			Aydin}, \bibinfo{person}{Quoc{-}Sang Phan}, \bibinfo{person}{Corina~S.
			Pasareanu}, {and} \bibinfo{person}{Tevfik Bultan}.}
	\bibinfo{year}{2016}\natexlab{}.
	\newblock \showarticletitle{String analysis for side channels with segmented
		oracles}. In \bibinfo{booktitle}{\emph{Proc. of the 24th {ACM} {SIGSOFT}
			International Symposium on Foundations of Software Engineering}}.
	\bibinfo{pages}{193--204}.
	\newblock
	\urldef\tempurl%
	\url{https://doi.org/10.1145/2950290.2950362}
	\showDOI{\tempurl}
	
	
	\bibitem[Bao et~al\mbox{.}(2021)]%
	{BaoWLLW21}
	\bibfield{author}{\bibinfo{person}{Qinkun Bao}, \bibinfo{person}{Zihao Wang},
		\bibinfo{person}{Xiaoting Li}, \bibinfo{person}{James~R. Larus}, {and}
		\bibinfo{person}{Dinghao Wu}.} \bibinfo{year}{2021}\natexlab{}.
	\newblock \showarticletitle{Abacus: Precise Side-Channel Analysis}. In
	\bibinfo{booktitle}{\emph{Proc. of the 43rd {IEEE/ACM} International
			Conference on Software Engineering}}. \bibinfo{pages}{797--809}.
	\newblock
	\urldef\tempurl%
	\url{https://doi.org/10.1109/ICSE43902.2021.00078}
	\showDOI{\tempurl}
	
	
	\bibitem[Barnett et~al\mbox{.}(2005)]%
	{BarnettCDJL05}
	\bibfield{author}{\bibinfo{person}{Michael Barnett},
		\bibinfo{person}{Bor{-}Yuh~Evan Chang}, \bibinfo{person}{Robert DeLine},
		\bibinfo{person}{Bart Jacobs}, {and} \bibinfo{person}{K.~Rustan~M. Leino}.}
	\bibinfo{year}{2005}\natexlab{}.
	\newblock \showarticletitle{Boogie: {A} Modular Reusable Verifier for
		Object-Oriented Programs}. In \bibinfo{booktitle}{\emph{Proc. of the 4th
			International Symposium on Formal Methods for Components and Objects}}.
	\bibinfo{pages}{364--387}.
	\newblock
	\urldef\tempurl%
	\url{https://doi.org/10.1007/11804192\_17}
	\showDOI{\tempurl}
	
	
	\bibitem[Barthe et~al\mbox{.}(2015)]%
	{BBDFGS15}
	\bibfield{author}{\bibinfo{person}{Gilles Barthe}, \bibinfo{person}{Sonia
			Bela{\"{\i}}d}, \bibinfo{person}{Fran{\c{c}}ois Dupressoir},
		\bibinfo{person}{Pierre{-}Alain Fouque}, \bibinfo{person}{Benjamin
			Gr{\'{e}}goire}, {and} \bibinfo{person}{Pierre{-}Yves Strub}.}
	\bibinfo{year}{2015}\natexlab{}.
	\newblock \showarticletitle{Verified Proofs of Higher-Order Masking}. In
	\bibinfo{booktitle}{\emph{Proc. of the 34th Annual International Conference
			on the Theory and Applications of Cryptographic Techniques}}.
	\bibinfo{pages}{457--485}.
	\newblock
	\urldef\tempurl%
	\url{https://doi.org/10.1007/978-3-662-46800-5\_18}
	\showDOI{\tempurl}
	
	
	\bibitem[Barthe et~al\mbox{.}(2014)]%
	{BBCLP14}
	\bibfield{author}{\bibinfo{person}{Gilles Barthe}, \bibinfo{person}{Gustavo
			Betarte}, \bibinfo{person}{Juan~Diego Campo}, \bibinfo{person}{Carlos~Daniel
			Luna}, {and} \bibinfo{person}{David Pichardie}.}
	\bibinfo{year}{2014}\natexlab{}.
	\newblock \showarticletitle{System-level Non-interference for Constant-time
		Cryptography}. In \bibinfo{booktitle}{\emph{Proc. of the 2014 {ACM} {SIGSAC}
			Conference on Computer and Communications Security}}.
	\bibinfo{pages}{1267--1279}.
	\newblock
	
	
	\bibitem[Barthe et~al\mbox{.}(2020)]%
	{BartheBGHLPT20}
	\bibfield{author}{\bibinfo{person}{Gilles Barthe}, \bibinfo{person}{Sandrine
			Blazy}, \bibinfo{person}{Benjamin Gr{\'{e}}goire},
		\bibinfo{person}{R{\'{e}}mi Hutin}, \bibinfo{person}{Vincent Laporte},
		\bibinfo{person}{David Pichardie}, {and} \bibinfo{person}{Alix Trieu}.}
	\bibinfo{year}{2020}\natexlab{}.
	\newblock \showarticletitle{Formal verification of a constant-time preserving
		{C} compiler}.
	\newblock \bibinfo{journal}{\emph{Proc. {ACM} Program. Lang.}}
	\bibinfo{volume}{4}, \bibinfo{number}{{POPL}} (\bibinfo{year}{2020}),
	\bibinfo{pages}{7:1--7:30}.
	\newblock
	\urldef\tempurl%
	\url{https://doi.org/10.1145/3371075}
	\showDOI{\tempurl}
	
	
	\bibitem[Bernstein et~al\mbox{.}(2008)]%
	{bernstein2008chacha}
	\bibfield{author}{\bibinfo{person}{Daniel~J Bernstein} {et~al\mbox{.}}}
	\bibinfo{year}{2008}\natexlab{}.
	\newblock \showarticletitle{ChaCha, a variant of Salsa20}. In
	\bibinfo{booktitle}{\emph{Workshop record of SASC}},
	Vol.~\bibinfo{volume}{8}. \bibinfo{pages}{3--5}.
	\newblock
	
	
	\bibitem[Bernstein et~al\mbox{.}(2012)]%
	{BernsteinLS12}
	\bibfield{author}{\bibinfo{person}{Daniel~J. Bernstein}, \bibinfo{person}{Tanja
			Lange}, {and} \bibinfo{person}{Peter Schwabe}.}
	\bibinfo{year}{2012}\natexlab{}.
	\newblock \showarticletitle{The Security Impact of a New Cryptographic
		Library}. In \bibinfo{booktitle}{\emph{Proc. of the 2nd International
			Conference on Cryptology and Information Security in Latin America}},
	Vol.~\bibinfo{volume}{7533}. \bibinfo{pages}{159--176}.
	\newblock
	\urldef\tempurl%
	\url{https://doi.org/10.1007/978-3-642-33481-8\_9}
	\showDOI{\tempurl}
	
	
	\bibitem[Beyer and Keremoglu(2011)]%
	{BeyerK11}
	\bibfield{author}{\bibinfo{person}{Dirk Beyer} {and} \bibinfo{person}{M.~Erkan
			Keremoglu}.} \bibinfo{year}{2011}\natexlab{}.
	\newblock \showarticletitle{CPAchecker: {A} Tool for Configurable Software
		Verification}. In \bibinfo{booktitle}{\emph{Proc. of the 23rd International
			Conference on Computer Aided Verification}}. \bibinfo{pages}{184--190}.
	\newblock
	\urldef\tempurl%
	\url{https://doi.org/10.1007/978-3-642-22110-1\_16}
	\showDOI{\tempurl}
	
	
	\bibitem[Blazy et~al\mbox{.}(2019)]%
	{BPT17}
	\bibfield{author}{\bibinfo{person}{Sandrine Blazy}, \bibinfo{person}{David
			Pichardie}, {and} \bibinfo{person}{Alix Trieu}.}
	\bibinfo{year}{2019}\natexlab{}.
	\newblock \showarticletitle{Verifying constant-time implementations by abstract
		interpretation}.
	\newblock \bibinfo{journal}{\emph{Journal of Computer Security}}
	\bibinfo{volume}{27}, \bibinfo{number}{1} (\bibinfo{year}{2019}),
	\bibinfo{pages}{137--163}.
	\newblock
	\urldef\tempurl%
	\url{https://doi.org/10.3233/JCS-181136}
	\showDOI{\tempurl}
	
	
	\bibitem[Borrello et~al\mbox{.}(2021)]%
	{BorrelloDQG21}
	\bibfield{author}{\bibinfo{person}{Pietro Borrello},
		\bibinfo{person}{Daniele~Cono D'Elia}, \bibinfo{person}{Leonardo Querzoni},
		{and} \bibinfo{person}{Cristiano Giuffrida}.}
	\bibinfo{year}{2021}\natexlab{}.
	\newblock \showarticletitle{Constantine: Automatic Side-Channel Resistance
		Using Efficient Control and Data Flow Linearization}. In
	\bibinfo{booktitle}{\emph{Proc. of the 2021 {ACM} {SIGSAC} Conference on
			Computer and Communications Security}}. \bibinfo{pages}{715--733}.
	\newblock
	\urldef\tempurl%
	\url{https://doi.org/10.1145/3460120.3484583}
	\showDOI{\tempurl}
	
	
	\bibitem[Brennan et~al\mbox{.}(2018)]%
	{BrennanSB18}
	\bibfield{author}{\bibinfo{person}{Tegan Brennan}, \bibinfo{person}{Seemanta
			Saha}, {and} \bibinfo{person}{Tevfik Bultan}.}
	\bibinfo{year}{2018}\natexlab{}.
	\newblock \showarticletitle{Symbolic Path Cost Analysis for Side-channel
		Detection}. In \bibinfo{booktitle}{\emph{Proc. of the 40th International
			Conference on Software Engineering: Companion Proceeedings}}.
	\bibinfo{pages}{424--425}.
	\newblock
	
	
	\bibitem[Brockschmidt et~al\mbox{.}(2016)]%
	{BrockschmidtCIK16}
	\bibfield{author}{\bibinfo{person}{Marc Brockschmidt}, \bibinfo{person}{Byron
			Cook}, \bibinfo{person}{Samin Ishtiaq}, \bibinfo{person}{Heidy Khlaaf}, {and}
		\bibinfo{person}{Nir Piterman}.} \bibinfo{year}{2016}\natexlab{}.
	\newblock \showarticletitle{{T2:} Temporal Property Verification}. In
	\bibinfo{booktitle}{\emph{Proc. of the 22nd International Conference on Tools
			and Algorithms for the Construction and Analysis of Systems}},
	Vol.~\bibinfo{volume}{9636}. \bibinfo{pages}{387--393}.
	\newblock
	\urldef\tempurl%
	\url{https://doi.org/10.1007/978-3-662-49674-9\_22}
	\showDOI{\tempurl}
	
	
	\bibitem[Brumley and Tuveri(2011)]%
	{BrumleyT11}
	\bibfield{author}{\bibinfo{person}{Billy~Bob Brumley} {and}
		\bibinfo{person}{Nicola Tuveri}.} \bibinfo{year}{2011}\natexlab{}.
	\newblock \showarticletitle{Remote Timing Attacks Are Still Practical}. In
	\bibinfo{booktitle}{\emph{Proc. of the 16th European Symposium on Research in
			Computer Security}}. \bibinfo{pages}{355--371}.
	\newblock
	
	
	\bibitem[Brumley and Boneh(2003)]%
	{brumley2003remote}
	\bibfield{author}{\bibinfo{person}{David Brumley} {and} \bibinfo{person}{Dan
			Boneh}.} \bibinfo{year}{2003}\natexlab{}.
	\newblock \showarticletitle{Remote Timing Attacks Are Practical}. In
	\bibinfo{booktitle}{\emph{Proc. of the 12th {USENIX} Security Symposium}}.
	\newblock
	
	
	\bibitem[Bultan(2019)]%
	{Bultan19}
	\bibfield{author}{\bibinfo{person}{Tevfik Bultan}.}
	\bibinfo{year}{2019}\natexlab{}.
	\newblock \showarticletitle{Quantifying Information Leakage Using Model
		Counting Constraint Solvers}. In \bibinfo{booktitle}{\emph{Proc. of the 11th
			International Conference on Verified Software: Theories, Tools, and
			Experiments}}, Vol.~\bibinfo{volume}{12031}. \bibinfo{pages}{30--35}.
	\newblock
	\urldef\tempurl%
	\url{https://doi.org/10.1007/978-3-030-41600-3\_3}
	\showDOI{\tempurl}
	
	
	\bibitem[Cauligi et~al\mbox{.}(2022)]%
	{CauligiDMBS22}
	\bibfield{author}{\bibinfo{person}{Sunjay Cauligi}, \bibinfo{person}{Craig
			Disselkoen}, \bibinfo{person}{Daniel Moghimi}, \bibinfo{person}{Gilles
			Barthe}, {and} \bibinfo{person}{Deian Stefan}.}
	\bibinfo{year}{2022}\natexlab{}.
	\newblock \showarticletitle{SoK: Practical Foundations for Software Spectre
		Defenses}. In \bibinfo{booktitle}{\emph{Proc. of the 43rd {IEEE} Symposium on
			Security and Privacy}}. \bibinfo{pages}{666--680}.
	\newblock
	
	
	\bibitem[Cauligi et~al\mbox{.}(2020)]%
	{CauligiDGTSRB20}
	\bibfield{author}{\bibinfo{person}{Sunjay Cauligi}, \bibinfo{person}{Craig
			Disselkoen}, \bibinfo{person}{Klaus von Gleissenthall},
		\bibinfo{person}{Dean~M. Tullsen}, \bibinfo{person}{Deian Stefan},
		\bibinfo{person}{Tamara Rezk}, {and} \bibinfo{person}{Gilles Barthe}.}
	\bibinfo{year}{2020}\natexlab{}.
	\newblock \showarticletitle{Constant-time foundations for the new spectre era}.
	In \bibinfo{booktitle}{\emph{Proc. of the 41st {ACM} {SIGPLAN} International
			Conference on Programming Language Design and Implementation}}.
	\bibinfo{pages}{913--926}.
	\newblock
	\urldef\tempurl%
	\url{https://doi.org/10.1145/3385412.3385970}
	\showDOI{\tempurl}
	
	
	\bibitem[Cauligi et~al\mbox{.}(2019)]%
	{CauligiSJBWRGBJ19}
	\bibfield{author}{\bibinfo{person}{Sunjay Cauligi}, \bibinfo{person}{Gary
			Soeller}, \bibinfo{person}{Brian Johannesmeyer}, \bibinfo{person}{Fraser
			Brown}, \bibinfo{person}{Riad~S. Wahby}, \bibinfo{person}{John Renner},
		\bibinfo{person}{Benjamin Gr{\'{e}}goire}, \bibinfo{person}{Gilles Barthe},
		\bibinfo{person}{Ranjit Jhala}, {and} \bibinfo{person}{Deian Stefan}.}
	\bibinfo{year}{2019}\natexlab{}.
	\newblock \showarticletitle{FaCT: a {DSL} for Timing-Sensitive Computation}. In
	\bibinfo{booktitle}{\emph{Proc. of the 40th {ACM} {SIGPLAN} International
			Conference on Programming Language Design and Implementation}}.
	\bibinfo{pages}{174--189}.
	\newblock
	\urldef\tempurl%
	\url{https://doi.org/10.1145/3314221.3314605}
	\showDOI{\tempurl}
	
	
	\bibitem[Daniel et~al\mbox{.}(2020)]%
	{DanielBR20}
	\bibfield{author}{\bibinfo{person}{Lesly{-}Ann Daniel},
		\bibinfo{person}{S{\'{e}}bastien Bardin}, {and} \bibinfo{person}{Tamara
			Rezk}.} \bibinfo{year}{2020}\natexlab{}.
	\newblock \showarticletitle{Binsec/Rel: Efficient Relational Symbolic Execution
		for Constant-Time at Binary-Level}. In \bibinfo{booktitle}{\emph{Proc. of the
			2020 {IEEE} Symposium on Security and Privacy}}. \bibinfo{pages}{1021--1038}.
	\newblock
	
	
	\bibitem[Daniel et~al\mbox{.}(2021)]%
	{DanielBR21}
	\bibfield{author}{\bibinfo{person}{Lesly{-}Ann Daniel},
		\bibinfo{person}{S{\'{e}}bastien Bardin}, {and} \bibinfo{person}{Tamara
			Rezk}.} \bibinfo{year}{2021}\natexlab{}.
	\newblock \showarticletitle{Hunting the Haunter - Efficient Relational Symbolic
		Execution for Spectre with Haunted RelSE}. In \bibinfo{booktitle}{\emph{Proc.
			of the 28th Annual Network and Distributed System Security Symposium}}.
	\newblock
	
	
	\bibitem[Daniel et~al\mbox{.}(2023a)]%
	{DanielBR23}
	\bibfield{author}{\bibinfo{person}{Lesly{-}Ann Daniel},
		\bibinfo{person}{S{\'{e}}bastien Bardin}, {and} \bibinfo{person}{Tamara
			Rezk}.} \bibinfo{year}{2023}\natexlab{a}.
	\newblock \showarticletitle{Binsec/Rel: Symbolic Binary Analyzer for Security
		with Applications to Constant-Time and Secret-Erasure}.
	\newblock \bibinfo{journal}{\emph{{ACM} Trans. Priv. Secur.}}
	\bibinfo{volume}{26}, \bibinfo{number}{2} (\bibinfo{year}{2023}),
	\bibinfo{pages}{11:1--11:42}.
	\newblock
	\urldef\tempurl%
	\url{https://doi.org/10.1145/3563037}
	\showDOI{\tempurl}
	
	
	\bibitem[Daniel et~al\mbox{.}(2023b)]%
	{DanielBNBRP23}
	\bibfield{author}{\bibinfo{person}{Lesly{-}Ann Daniel}, \bibinfo{person}{Marton
			Bognar}, \bibinfo{person}{Job Noorman}, \bibinfo{person}{S{\'{e}}bastien
			Bardin}, \bibinfo{person}{Tamara Rezk}, {and} \bibinfo{person}{Frank
			Piessens}.} \bibinfo{year}{2023}\natexlab{b}.
	\newblock \showarticletitle{ProSpeCT: Provably Secure Speculation for the
		Constant-Time Policy}. In \bibinfo{booktitle}{\emph{Proc. of the 32nd
			{USENIX} Security Symposium}}.
	\newblock
	
	
	\bibitem[Doychev et~al\mbox{.}(2013)]%
	{DoychevFKMR13}
	\bibfield{author}{\bibinfo{person}{Goran Doychev}, \bibinfo{person}{Dominik
			Feld}, \bibinfo{person}{Boris K{\"{o}}pf}, \bibinfo{person}{Laurent
			Mauborgne}, {and} \bibinfo{person}{Jan Reineke}.}
	\bibinfo{year}{2013}\natexlab{}.
	\newblock \showarticletitle{CacheAudit: {A} Tool for the Static Analysis of
		Cache Side Channels}. In \bibinfo{booktitle}{\emph{Proc. of the 22th {USENIX}
			Security Symposium}}. \bibinfo{pages}{431--446}.
	\newblock
	
	
	\bibitem[Fan et~al\mbox{.}(2022)]%
	{FanSCZL22}
	\bibfield{author}{\bibinfo{person}{Yuxin Fan}, \bibinfo{person}{Fu Song},
		\bibinfo{person}{Taolue Chen}, \bibinfo{person}{Liangfeng Zhang}, {and}
		\bibinfo{person}{Wanwei Liu}.} \bibinfo{year}{2022}\natexlab{}.
	\newblock \showarticletitle{PoS4MPC: Automated Security Policy Synthesis for
		Secure Multi-party Computation}. In \bibinfo{booktitle}{\emph{Proc. of the
			34th International Conference on Computer Aided Verification}},
	\bibfield{editor}{\bibinfo{person}{Sharon Shoham} {and}
		\bibinfo{person}{Yakir Vizel}} (Eds.), Vol.~\bibinfo{volume}{13371}.
	\bibinfo{pages}{385--406}.
	\newblock
	\urldef\tempurl%
	\url{https://doi.org/10.1007/978-3-031-13185-1\_19}
	\showDOI{\tempurl}
	
	
	\bibitem[Farina et~al\mbox{.}(2019)]%
	{FarinaCG19}
	\bibfield{author}{\bibinfo{person}{Gian~Pietro Farina},
		\bibinfo{person}{Stephen Chong}, {and} \bibinfo{person}{Marco Gaboardi}.}
	\bibinfo{year}{2019}\natexlab{}.
	\newblock \showarticletitle{Relational Symbolic Execution}. In
	\bibinfo{booktitle}{\emph{Proc. of the 21st International Symposium on
			Principles and Practice of Programming Languages}}.
	\bibinfo{pages}{10:1--10:14}.
	\newblock
	\urldef\tempurl%
	\url{https://doi.org/10.1145/3354166.3354175}
	\showDOI{\tempurl}
	
	
	\bibitem[Gao et~al\mbox{.}(2021)]%
	{GaoXSC2020}
	\bibfield{author}{\bibinfo{person}{Pengfei Gao}, \bibinfo{person}{Hongyi Xie},
		\bibinfo{person}{Fu Song}, {and} \bibinfo{person}{Taolue Chen}.}
	\bibinfo{year}{2021}\natexlab{}.
	\newblock \showarticletitle{A Hybrid Approach to Formal Verification of
		Higher-Order Masked Arithmetic Programs}.
	\newblock \bibinfo{journal}{\emph{{ACM} Trans. Softw. Eng. Methodol.}}
	\bibinfo{volume}{30}, \bibinfo{number}{3} (\bibinfo{year}{2021}),
	\bibinfo{pages}{26:1--26:42}.
	\newblock
	\urldef\tempurl%
	\url{https://doi.org/10.1145/3428015}
	\showDOI{\tempurl}
	
	
	\bibitem[Gao et~al\mbox{.}(2022)]%
	{GXSZSC2020}
	\bibfield{author}{\bibinfo{person}{Pengfei Gao}, \bibinfo{person}{Hongyi Xie},
		\bibinfo{person}{Pu Sun}, \bibinfo{person}{Jun Zhang}, \bibinfo{person}{Fu
			Song}, {and} \bibinfo{person}{Taolue Chen}.} \bibinfo{year}{2022}\natexlab{}.
	\newblock \showarticletitle{Formal Verification of Masking Countermeasures for
		Arithmetic Programs}.
	\newblock \bibinfo{journal}{\emph{{IEEE} Trans. Software Eng.}}
	\bibinfo{volume}{48}, \bibinfo{number}{3} (\bibinfo{year}{2022}),
	\bibinfo{pages}{973--1000}.
	\newblock
	\urldef\tempurl%
	\url{https://doi.org/10.1109/TSE.2020.3008852}
	\showDOI{\tempurl}
	
	
	\bibitem[Gao et~al\mbox{.}(2019a)]%
	{GXZSC19}
	\bibfield{author}{\bibinfo{person}{Pengfei Gao}, \bibinfo{person}{Hongyi Xie},
		\bibinfo{person}{Jun Zhang}, \bibinfo{person}{Fu Song}, {and}
		\bibinfo{person}{Taolue Chen}.} \bibinfo{year}{2019}\natexlab{a}.
	\newblock \showarticletitle{Quantitative Verification of Masked Arithmetic
		Programs Against Side-Channel Attacks}. In \bibinfo{booktitle}{\emph{Proc. of
			the 25th International Conference on Tools and Algorithms for the
			Construction and Analysis of Systems}}. \bibinfo{pages}{155--173}.
	\newblock
	\urldef\tempurl%
	\url{https://doi.org/10.1007/978-3-030-17462-0\_9}
	\showDOI{\tempurl}
	
	
	\bibitem[Gao et~al\mbox{.}(2019b)]%
	{GZSC19}
	\bibfield{author}{\bibinfo{person}{Pengfei Gao}, \bibinfo{person}{Jun Zhang},
		\bibinfo{person}{Fu Song}, {and} \bibinfo{person}{Chao Wang}.}
	\bibinfo{year}{2019}\natexlab{b}.
	\newblock \showarticletitle{Verifying and Quantifying Side-channel Resistance
		of Masked Software Implementations}.
	\newblock \bibinfo{journal}{\emph{{ACM} Trans. Softw. Eng. Methodol.}}
	\bibinfo{volume}{28}, \bibinfo{number}{3} (\bibinfo{year}{2019}),
	\bibinfo{pages}{16:1--16:32}.
	\newblock
	\urldef\tempurl%
	\url{https://doi.org/10.1145/3330392}
	\showDOI{\tempurl}
	
	
	\bibitem[Gao et~al\mbox{.}(2023)]%
	{GaoZSCS23}
	\bibfield{author}{\bibinfo{person}{Pengfei Gao}, \bibinfo{person}{Yedi Zhang},
		\bibinfo{person}{Fu Song}, \bibinfo{person}{Taolue Chen}, {and}
		\bibinfo{person}{Fran{\c{c}}ois{-}Xavier Standaert}.}
	\bibinfo{year}{2023}\natexlab{}.
	\newblock \showarticletitle{Compositional Verification of Efficient Masking
		Countermeasures against Side-Channel Attacks}.
	\newblock \bibinfo{journal}{\emph{Proc. {ACM} Program. Lang.}}
	\bibinfo{volume}{7}, \bibinfo{number}{{OOPSLA2}} (\bibinfo{year}{2023}),
	\bibinfo{pages}{1817--1847}.
	\newblock
	\urldef\tempurl%
	\url{https://doi.org/10.1145/3622862}
	\showDOI{\tempurl}
	
	
	\bibitem[Geimer et~al\mbox{.}(2023)]%
	{GeimerVRDBM23}
	\bibfield{author}{\bibinfo{person}{Antoine Geimer},
		\bibinfo{person}{Math{\'{e}}o Vergnolle},
		\bibinfo{person}{Fr{\'{e}}d{\'{e}}ric Recoules}, \bibinfo{person}{Lesly{-}Ann
			Daniel}, \bibinfo{person}{S{\'{e}}bastien Bardin}, {and}
		\bibinfo{person}{Cl{\'{e}}mentine Maurice}.} \bibinfo{year}{2023}\natexlab{}.
	\newblock \showarticletitle{A Systematic Evaluation of Automated Tools for
		Side-Channel Vulnerabilities Detection in Cryptographic Libraries}. In
	\bibinfo{booktitle}{\emph{Proc. of the 2023 {ACM} {SIGSAC} Conference on
			Computer and Communications Security}},
	\bibfield{editor}{\bibinfo{person}{Weizhi Meng},
		\bibinfo{person}{Christian~Damsgaard Jensen}, \bibinfo{person}{Cas Cremers},
		{and} \bibinfo{person}{Engin Kirda}} (Eds.). \bibinfo{publisher}{{ACM}},
	\bibinfo{pages}{1690--1704}.
	\newblock
	\urldef\tempurl%
	\url{https://doi.org/10.1145/3576915.3623112}
	\showDOI{\tempurl}
	
	
	\bibitem[Guarnieri et~al\mbox{.}(2020)]%
	{GuarnieriKMRS20}
	\bibfield{author}{\bibinfo{person}{Marco Guarnieri}, \bibinfo{person}{Boris
			K{\"{o}}pf}, \bibinfo{person}{Jos{\'{e}}~F. Morales}, \bibinfo{person}{Jan
			Reineke}, {and} \bibinfo{person}{Andr{\'{e}}s S{\'{a}}nchez}.}
	\bibinfo{year}{2020}\natexlab{}.
	\newblock \showarticletitle{Spectector: Principled Detection of Speculative
		Information Flows}. In \bibinfo{booktitle}{\emph{Proc. of the {IEEE}
			Symposium on Security and Privacy}}. \bibinfo{pages}{1--19}.
	\newblock
	\urldef\tempurl%
	\url{https://doi.org/10.1109/SP40000.2020.00011}
	\showDOI{\tempurl}
	
	
	\bibitem[Guo et~al\mbox{.}(2020)]%
	{GuoCLCWW020}
	\bibfield{author}{\bibinfo{person}{Shengjian Guo}, \bibinfo{person}{Yueqi
			Chen}, \bibinfo{person}{Peng Li}, \bibinfo{person}{Yueqiang Cheng},
		\bibinfo{person}{Huibo Wang}, \bibinfo{person}{Meng Wu}, {and}
		\bibinfo{person}{Zhiqiang Zuo}.} \bibinfo{year}{2020}\natexlab{}.
	\newblock \showarticletitle{SpecuSym: speculative symbolic execution for cache
		timing leak detection}. In \bibinfo{booktitle}{\emph{Proc. of the 42nd
			International Conference on Software Engineering}}.
	\bibinfo{pages}{1235--1247}.
	\newblock
	\urldef\tempurl%
	\url{https://doi.org/10.1145/3377811.3380428}
	\showDOI{\tempurl}
	
	
	\bibitem[He et~al\mbox{.}(2020)]%
	{HeEC20}
	\bibfield{author}{\bibinfo{person}{Shaobo He}, \bibinfo{person}{Michael Emmi},
		{and} \bibinfo{person}{Gabriela~F. Ciocarlie}.}
	\bibinfo{year}{2020}\natexlab{}.
	\newblock \showarticletitle{ct-fuzz: Fuzzing for Timing Leaks}. In
	\bibinfo{booktitle}{\emph{Proc. of the 13th {IEEE} International Conference
			on Software Testing, Validation and Verification}}.
	\bibinfo{pages}{466--471}.
	\newblock
	\urldef\tempurl%
	\url{https://doi.org/10.1109/ICST46399.2020.00063}
	\showDOI{\tempurl}
	
	
	\bibitem[Jiang et~al\mbox{.}(2022)]%
	{JiangBWL022}
	\bibfield{author}{\bibinfo{person}{Ke Jiang}, \bibinfo{person}{Yuyan Bao},
		\bibinfo{person}{Shuai Wang}, \bibinfo{person}{Zhibo Liu}, {and}
		\bibinfo{person}{Tianwei Zhang}.} \bibinfo{year}{2022}\natexlab{}.
	\newblock \showarticletitle{Cache Refinement Type for Side-Channel Detection of
		Cryptographic Software}. In \bibinfo{booktitle}{\emph{Proc. of the 2022 {ACM}
			{SIGSAC} Conference on Computer and Communications Security}}.
	\bibinfo{pages}{1583--1597}.
	\newblock
	\urldef\tempurl%
	\url{https://doi.org/10.1145/3548606.3560672}
	\showDOI{\tempurl}
	
	
	\bibitem[Kocher et~al\mbox{.}(2019)]%
	{KocherHFGGHHLM019}
	\bibfield{author}{\bibinfo{person}{Paul Kocher}, \bibinfo{person}{Jann Horn},
		\bibinfo{person}{Anders Fogh}, \bibinfo{person}{Daniel Genkin},
		\bibinfo{person}{Daniel Gruss}, \bibinfo{person}{Werner Haas},
		\bibinfo{person}{Mike Hamburg}, \bibinfo{person}{Moritz Lipp},
		\bibinfo{person}{Stefan Mangard}, \bibinfo{person}{Thomas Prescher},
		\bibinfo{person}{Michael Schwarz}, {and} \bibinfo{person}{Yuval Yarom}.}
	\bibinfo{year}{2019}\natexlab{}.
	\newblock \showarticletitle{Spectre Attacks: Exploiting Speculative Execution}.
	In \bibinfo{booktitle}{\emph{Proc. of {IEEE} Symposium on Security and
			Privacy}}. \bibinfo{pages}{1--19}.
	\newblock
	
	
	\bibitem[Kocher(1996)]%
	{Kocher96}
	\bibfield{author}{\bibinfo{person}{Paul~C. Kocher}.}
	\bibinfo{year}{1996}\natexlab{}.
	\newblock \showarticletitle{Timing Attacks on Implementations of
		Diffie-Hellman, RSA, DSS, and Other Systems}. In
	\bibinfo{booktitle}{\emph{Proc. of the 16th Annual International Cryptology
			Conference on Advances in Cryptology}}. \bibinfo{pages}{104--113}.
	\newblock
	
	
	\bibitem[Li et~al\mbox{.}(2021)]%
	{li2021scaling}
	\bibfield{author}{\bibinfo{person}{Haofeng Li}, \bibinfo{person}{Haining Meng},
		\bibinfo{person}{Hengjie Zheng}, \bibinfo{person}{Liqing Cao},
		\bibinfo{person}{Jie Lu}, \bibinfo{person}{Lian Li}, {and}
		\bibinfo{person}{Lin Gao}.} \bibinfo{year}{2021}\natexlab{}.
	\newblock \showarticletitle{Scaling up the IFDS algorithm with efficient
		disk-assisted computing}. In \bibinfo{booktitle}{\emph{Proc. of the IEEE/ACM
			International Symposium on Code Generation and Optimization}}.
	\bibinfo{pages}{236--247}.
	\newblock
	
	
	\bibitem[Li et~al\mbox{.}(2014)]%
	{LiGR14}
	\bibfield{author}{\bibinfo{person}{Peng Li}, \bibinfo{person}{Debin Gao}, {and}
		\bibinfo{person}{Michael~K. Reiter}.} \bibinfo{year}{2014}\natexlab{}.
	\newblock \showarticletitle{StopWatch: {A} Cloud Architecture for Timing
		Channel Mitigation}.
	\newblock \bibinfo{journal}{\emph{{ACM} Trans. Inf. Syst. Secur.}}
	\bibinfo{volume}{17}, \bibinfo{number}{2} (\bibinfo{year}{2014}),
	\bibinfo{pages}{8:1--8:28}.
	\newblock
	
	
	\bibitem[Lipp et~al\mbox{.}(2018)]%
	{Lipp0G0HFHMKGYH18}
	\bibfield{author}{\bibinfo{person}{Moritz Lipp}, \bibinfo{person}{Michael
			Schwarz}, \bibinfo{person}{Daniel Gruss}, \bibinfo{person}{Thomas Prescher},
		\bibinfo{person}{Werner Haas}, \bibinfo{person}{Anders Fogh},
		\bibinfo{person}{Jann Horn}, \bibinfo{person}{Stefan Mangard},
		\bibinfo{person}{Paul Kocher}, \bibinfo{person}{Daniel Genkin},
		\bibinfo{person}{Yuval Yarom}, {and} \bibinfo{person}{Mike Hamburg}.}
	\bibinfo{year}{2018}\natexlab{}.
	\newblock \showarticletitle{Meltdown: Reading Kernel Memory from User Space}.
	In \bibinfo{booktitle}{\emph{Proc. of the 27th {USENIX} Security Symposium}}.
	\bibinfo{pages}{973--990}.
	\newblock
	
	
	\bibitem[{Michael Emmi}(2023)]%
	{Bam-bam-boogieman}
	\bibfield{author}{\bibinfo{person}{{Michael Emmi}}.}
	\bibinfo{year}{2023}\natexlab{}.
	\newblock \bibinfo{booktitle}{\emph{{bam-bam-boogieman}}}.
	\newblock
	\newblock
	\shownote{\url{https://github.com/michael-emmi/bam-bam-boogieman}}.
	
	
	\bibitem[{Microsoft}(2023)]%
	{FourQlib}
	\bibfield{author}{\bibinfo{person}{{Microsoft}}.}
	\bibinfo{year}{2023}\natexlab{}.
	\newblock \bibinfo{booktitle}{\emph{{FourQlib}}}.
	\newblock
	\newblock
	\shownote{\url{https://github.com/microsoft/FourQlib}}.
	
	
	\bibitem[Oh et~al\mbox{.}(2012)]%
	{OhHLLY12}
	\bibfield{author}{\bibinfo{person}{Hakjoo Oh}, \bibinfo{person}{Kihong Heo},
		\bibinfo{person}{Wonchan Lee}, \bibinfo{person}{Woosuk Lee}, {and}
		\bibinfo{person}{Kwangkeun Yi}.} \bibinfo{year}{2012}\natexlab{}.
	\newblock \showarticletitle{Design and implementation of sparse global analyses
		for C-like languages}. In \bibinfo{booktitle}{\emph{Proc. of the {ACM}
			{SIGPLAN} Conference on Programming Language Design and Implementation}}.
	\bibinfo{pages}{229--238}.
	\newblock
	
	
	\bibitem[Pornin(2023)]%
	{BearSSL23}
	\bibfield{author}{\bibinfo{person}{Thomas Pornin}.}
	\bibinfo{year}{2023}\natexlab{}.
	\newblock \bibinfo{booktitle}{\emph{{BearSSL}}}.
	\newblock
	\newblock
	\shownote{\url{https://bearssl.org}}.
	
	
	\bibitem[Qin et~al\mbox{.}(2022)]%
	{QinJSCX22}
	\bibfield{author}{\bibinfo{person}{Qi Qin}, \bibinfo{person}{JulianAndres
			JiYang}, \bibinfo{person}{Fu Song}, \bibinfo{person}{Taolue Chen}, {and}
		\bibinfo{person}{Xinyu Xing}.} \bibinfo{year}{2022}\natexlab{}.
	\newblock \showarticletitle{DeJITLeak: eliminating JIT-induced timing
		side-channel leaks}. In \bibinfo{booktitle}{\emph{Proc. of the 30th {ACM}
			Joint European Software Engineering Conference and Symposium on the
			Foundations of Software Engineering}}. \bibinfo{pages}{872--884}.
	\newblock
	\urldef\tempurl%
	\url{https://doi.org/10.1145/3540250.3549150}
	\showDOI{\tempurl}
	
	
	\bibitem[Rakamari{\'c} and Emmi(2014)]%
	{rakamaric2014smack}
	\bibfield{author}{\bibinfo{person}{Zvonimir Rakamari{\'c}} {and}
		\bibinfo{person}{Michael Emmi}.} \bibinfo{year}{2014}\natexlab{}.
	\newblock \showarticletitle{SMACK: Decoupling source language details from
		verifier implementations}. In \bibinfo{booktitle}{\emph{Proc. of the 26th
			International Conference on Computer Aided Verification}}.
	\bibinfo{pages}{106--113}.
	\newblock
	\urldef\tempurl%
	\url{https://doi.org/10.1007/978-3-319-08867-9\_7}
	\showDOI{\tempurl}
	
	
	\bibitem[Reps et~al\mbox{.}(1995)]%
	{RHS95}
	\bibfield{author}{\bibinfo{person}{Thomas~W. Reps}, \bibinfo{person}{Susan
			Horwitz}, {and} \bibinfo{person}{Shmuel Sagiv}.}
	\bibinfo{year}{1995}\natexlab{}.
	\newblock \showarticletitle{Precise Interprocedural Dataflow Analysis via Graph
		Reachability}. In \bibinfo{booktitle}{\emph{Proc. of the 22nd {ACM}
			{SIGPLAN-SIGACT} Symposium on Principles of Programming Languages}}.
	\bibinfo{pages}{49--61}.
	\newblock
	
	
	\bibitem[Saha et~al\mbox{.}(2023)]%
	{SahaGLBB23}
	\bibfield{author}{\bibinfo{person}{Seemanta Saha}, \bibinfo{person}{Surendra
			Ghentiyala}, \bibinfo{person}{Shihua Lu}, \bibinfo{person}{Lucas Bang}, {and}
		\bibinfo{person}{Tevfik Bultan}.} \bibinfo{year}{2023}\natexlab{}.
	\newblock \showarticletitle{Obtaining Information Leakage Bounds via
		Approximate Model Counting}.
	\newblock \bibinfo{journal}{\emph{Proc. {ACM} Program. Lang.}}
	\bibinfo{volume}{7}, \bibinfo{number}{{PLDI}} (\bibinfo{year}{2023}),
	\bibinfo{pages}{1488--1509}.
	\newblock
	\urldef\tempurl%
	\url{https://doi.org/10.1145/3591281}
	\showDOI{\tempurl}
	
	
	\bibitem[Schubert et~al\mbox{.}(2019)]%
	{schubert2019phasar}
	\bibfield{author}{\bibinfo{person}{Philipp~Dominik Schubert},
		\bibinfo{person}{Ben Hermann}, {and} \bibinfo{person}{Eric Bodden}.}
	\bibinfo{year}{2019}\natexlab{}.
	\newblock \showarticletitle{PhASAR: An Inter-procedural Static Analysis
		Framework for {C/C++}}. In \bibinfo{booktitle}{\emph{Proc. of the 25th
			International Conference on Tools and Algorithms for the Construction and
			Analysis of Systems}}. \bibinfo{pages}{393--410}.
	\newblock
	\urldef\tempurl%
	\url{https://doi.org/10.1007/978-3-030-17465-1\_22}
	\showDOI{\tempurl}
	
	
	\bibitem[Services(2023)]%
	{s2nctverif}
	\bibfield{author}{\bibinfo{person}{Amazon~Web Services}.}
	\bibinfo{year}{2023}\natexlab{}.
	\newblock \bibinfo{booktitle}{\emph{{ct-verif for s2n}}}.
	\newblock
	\newblock
	\shownote{\url{https://github.com/aws/s2n-tls/tree/main/tests/ctverif}}.
	
	
	\bibitem[Sousa and Dillig(2016)]%
	{SousaD16}
	\bibfield{author}{\bibinfo{person}{Marcelo Sousa} {and} \bibinfo{person}{Isil
			Dillig}.} \bibinfo{year}{2016}\natexlab{}.
	\newblock \showarticletitle{Cartesian hoare logic for verifying k-safety
		properties}. In \bibinfo{booktitle}{\emph{Proc. of the 37th {ACM} {SIGPLAN}
			Conference on Programming Language Design and Implementation}}.
	\bibinfo{pages}{57--69}.
	\newblock
	
	
	\bibitem[Sui and Xue(2016)]%
	{sui2016svf}
	\bibfield{author}{\bibinfo{person}{Yulei Sui} {and} \bibinfo{person}{Jingling
			Xue}.} \bibinfo{year}{2016}\natexlab{}.
	\newblock \showarticletitle{SVF: interprocedural static value-flow analysis in
		LLVM}. In \bibinfo{booktitle}{\emph{Proc. of the 25th international
			conference on compiler construction}}. \bibinfo{pages}{265--266}.
	\newblock
	\urldef\tempurl%
	\url{https://doi.org/10.1145/2892208.2892235}
	\showDOI{\tempurl}
	
	
	\bibitem[Sung et~al\mbox{.}(2018)]%
	{SungPW18}
	\bibfield{author}{\bibinfo{person}{Chungha Sung}, \bibinfo{person}{Brandon
			Paulsen}, {and} \bibinfo{person}{Chao Wang}.}
	\bibinfo{year}{2018}\natexlab{}.
	\newblock \showarticletitle{{CANAL:} a cache timing analysis framework via
		{LLVM} transformation}. In \bibinfo{booktitle}{\emph{Proc. of the 33rd
			{ACM/IEEE} International Conference on Automated Software Engineering}}.
	\bibinfo{pages}{904--907}.
	\newblock
	\urldef\tempurl%
	\url{https://doi.org/10.1145/3238147.3240485}
	\showDOI{\tempurl}
	
	
	\bibitem[Terauchi and Aiken(2005)]%
	{TerauchiA05}
	\bibfield{author}{\bibinfo{person}{Tachio Terauchi} {and}
		\bibinfo{person}{Alexander Aiken}.} \bibinfo{year}{2005}\natexlab{}.
	\newblock \showarticletitle{Secure Information Flow as a Safety Problem}. In
	\bibinfo{booktitle}{\emph{Proc. of the 12th International Symposium on Static
			Analysis}}. \bibinfo{pages}{352--367}.
	\newblock
	\urldef\tempurl%
	\url{https://doi.org/10.1007/11547662_24}
	\showDOI{\tempurl}
	
	
	\bibitem[{Trusted Firmware Project}(2023)]%
	{MbedTLS23}
	\bibfield{author}{\bibinfo{person}{{Trusted Firmware Project}}.}
	\bibinfo{year}{2023}\natexlab{}.
	\newblock \bibinfo{booktitle}{\emph{{Mbed TLS}}}.
	\newblock
	\newblock
	\shownote{\url{https://github.com/Mbed-TLS/mbedtls}}.
	
	
	\bibitem[Vassena et~al\mbox{.}(2021)]%
	{VassenaDGCKJTS21}
	\bibfield{author}{\bibinfo{person}{Marco Vassena}, \bibinfo{person}{Craig
			Disselkoen}, \bibinfo{person}{Klaus von Gleissenthall},
		\bibinfo{person}{Sunjay Cauligi}, \bibinfo{person}{Rami~G{\"{o}}khan Kici},
		\bibinfo{person}{Ranjit Jhala}, \bibinfo{person}{Dean~M. Tullsen}, {and}
		\bibinfo{person}{Deian Stefan}.} \bibinfo{year}{2021}\natexlab{}.
	\newblock \showarticletitle{Automatically eliminating speculative leaks from
		cryptographic code with blade}.
	\newblock \bibinfo{journal}{\emph{Proc. {ACM} Program. Lang.}}
	\bibinfo{volume}{5}, \bibinfo{number}{{POPL}} (\bibinfo{year}{2021}),
	\bibinfo{pages}{1--30}.
	\newblock
	\urldef\tempurl%
	\url{https://doi.org/10.1145/3434330}
	\showDOI{\tempurl}
	
	
	\bibitem[Wang et~al\mbox{.}(2020)]%
	{WangCBMR20}
	\bibfield{author}{\bibinfo{person}{Guanhua Wang}, \bibinfo{person}{Sudipta
			Chattopadhyay}, \bibinfo{person}{Arnab~Kumar Biswas}, \bibinfo{person}{Tulika
			Mitra}, {and} \bibinfo{person}{Abhik Roychoudhury}.}
	\bibinfo{year}{2020}\natexlab{}.
	\newblock \showarticletitle{KLEESpectre: Detecting Information Leakage through
		Speculative Cache Attacks via Symbolic Execution}.
	\newblock \bibinfo{journal}{\emph{{ACM} Trans. Softw. Eng. Methodol.}}
	\bibinfo{volume}{29}, \bibinfo{number}{3} (\bibinfo{year}{2020}),
	\bibinfo{pages}{14:1--14:31}.
	\newblock
	\urldef\tempurl%
	\url{https://doi.org/10.1145/3385897}
	\showDOI{\tempurl}
	
	
	\bibitem[Wang et~al\mbox{.}(2021)]%
	{WangCGMR21}
	\bibfield{author}{\bibinfo{person}{Guanhua Wang}, \bibinfo{person}{Sudipta
			Chattopadhyay}, \bibinfo{person}{Ivan Gotovchits}, \bibinfo{person}{Tulika
			Mitra}, {and} \bibinfo{person}{Abhik Roychoudhury}.}
	\bibinfo{year}{2021}\natexlab{}.
	\newblock \showarticletitle{oo7: Low-Overhead Defense Against Spectre Attacks
		via Program Analysis}.
	\newblock \bibinfo{journal}{\emph{{IEEE} Trans. Software Eng.}}
	\bibinfo{volume}{47}, \bibinfo{number}{11} (\bibinfo{year}{2021}),
	\bibinfo{pages}{2504--2519}.
	\newblock
	\urldef\tempurl%
	\url{https://doi.org/10.1109/TSE.2019.2953709}
	\showDOI{\tempurl}
	
	
	\bibitem[Wang et~al\mbox{.}(2019)]%
	{WBLWZWu19}
	\bibfield{author}{\bibinfo{person}{Shuai Wang}, \bibinfo{person}{Yuyan Bao},
		\bibinfo{person}{Xiao Liu}, \bibinfo{person}{Pei Wang},
		\bibinfo{person}{Danfeng Zhang}, {and} \bibinfo{person}{Dinghao Wu}.}
	\bibinfo{year}{2019}\natexlab{}.
	\newblock \showarticletitle{Identifying {Cache-Based} Side Channels through
		{Secret-Augmented} Abstract Interpretation}. In
	\bibinfo{booktitle}{\emph{Proc. of the 28th {USENIX} Security Symposium}}.
	\bibinfo{pages}{657--674}.
	\newblock
	
	
	\bibitem[Watt et~al\mbox{.}(2019)]%
	{WattRPCS19}
	\bibfield{author}{\bibinfo{person}{Conrad Watt}, \bibinfo{person}{John Renner},
		\bibinfo{person}{Natalie Popescu}, \bibinfo{person}{Sunjay Cauligi}, {and}
		\bibinfo{person}{Deian Stefan}.} \bibinfo{year}{2019}\natexlab{}.
	\newblock \showarticletitle{CT-Wasm: Type-Driven Secure Cryptography for the
		Web Ecosystem}.
	\newblock \bibinfo{journal}{\emph{Proc. of the {ACM} on Programming Languages}}
	\bibinfo{volume}{3}, \bibinfo{number}{{POPL}} (\bibinfo{year}{2019}),
	\bibinfo{pages}{77:1--77:29}.
	\newblock
	\urldef\tempurl%
	\url{https://doi.org/10.1145/3290390}
	\showDOI{\tempurl}
	
	
	\bibitem[Weiser et~al\mbox{.}(2018)]%
	{WeiserZSMMS18}
	\bibfield{author}{\bibinfo{person}{Samuel Weiser}, \bibinfo{person}{Andreas
			Zankl}, \bibinfo{person}{Raphael Spreitzer}, \bibinfo{person}{Katja Miller},
		\bibinfo{person}{Stefan Mangard}, {and} \bibinfo{person}{Georg Sigl}.}
	\bibinfo{year}{2018}\natexlab{}.
	\newblock \showarticletitle{{DATA} - Differential Address Trace Analysis:
		Finding Address-based Side-Channels in Binaries}. In
	\bibinfo{booktitle}{\emph{Proc. of the 27th {USENIX} Security Symposium}}.
	\bibinfo{pages}{603--620}.
	\newblock
	
	
	\bibitem[Wu et~al\mbox{.}(2018)]%
	{WuGS018}
	\bibfield{author}{\bibinfo{person}{Meng Wu}, \bibinfo{person}{Shengjian Guo},
		\bibinfo{person}{Patrick Schaumont}, {and} \bibinfo{person}{Chao Wang}.}
	\bibinfo{year}{2018}\natexlab{}.
	\newblock \showarticletitle{Eliminating Timing Side-Channel Leaks using Program
		Repair}. In \bibinfo{booktitle}{\emph{Proc. of the 27th {ACM} {SIGSOFT}
			International Symposium on Software Testing and Analysis}}.
	\bibinfo{pages}{15--26}.
	\newblock
	\urldef\tempurl%
	\url{https://doi.org/10.1145/3213846.3213851}
	\showDOI{\tempurl}
	
	
	\bibitem[Wu and Wang(2019)]%
	{Wu019}
	\bibfield{author}{\bibinfo{person}{Meng Wu} {and} \bibinfo{person}{Chao Wang}.}
	\bibinfo{year}{2019}\natexlab{}.
	\newblock \showarticletitle{Abstract interpretation under speculative
		execution}. In \bibinfo{booktitle}{\emph{Proc. of the 40th {ACM} {SIGPLAN}
			Conference on Programming Language Design and Implementation}}.
	\bibinfo{pages}{802--815}.
	\newblock
	\urldef\tempurl%
	\url{https://doi.org/10.1145/3314221.3314647}
	\showDOI{\tempurl}
	
	
	\bibitem[Yang et~al\mbox{.}(2018)]%
	{YangVSGM18}
	\bibfield{author}{\bibinfo{person}{Weikun Yang}, \bibinfo{person}{Yakir Vizel},
		\bibinfo{person}{Pramod Subramanyan}, \bibinfo{person}{Aarti Gupta}, {and}
		\bibinfo{person}{Sharad Malik}.} \bibinfo{year}{2018}\natexlab{}.
	\newblock \showarticletitle{Lazy Self-composition for Security Verification}.
	In \bibinfo{booktitle}{\emph{Proc. of the 30th International Conference on
			Computer Aided Verification}}. \bibinfo{pages}{136--156}.
	\newblock
	\urldef\tempurl%
	\url{https://doi.org/10.1007/978-3-319-96142-2\_11}
	\showDOI{\tempurl}
	
	
	\bibitem[Yavuz et~al\mbox{.}(2023)]%
	{YavuzFHBBT23}
	\bibfield{author}{\bibinfo{person}{Tuba Yavuz}, \bibinfo{person}{Farhaan
			Fowze}, \bibinfo{person}{Grant Hernandez}, \bibinfo{person}{Ken~Yihang Bai},
		\bibinfo{person}{Kevin R.~B. Butler}, {and} \bibinfo{person}{Dave~Jing
			Tian}.} \bibinfo{year}{2023}\natexlab{}.
	\newblock \showarticletitle{{ENCIDER:} Detecting Timing and Cache Side Channels
		in {SGX} Enclaves and Cryptographic APIs}.
	\newblock \bibinfo{journal}{\emph{{IEEE} Trans. Dependable Secur. Comput.}}
	\bibinfo{volume}{20}, \bibinfo{number}{2} (\bibinfo{year}{2023}),
	\bibinfo{pages}{1577--1595}.
	\newblock
	\urldef\tempurl%
	\url{https://doi.org/10.1109/TDSC.2022.3160346}
	\showDOI{\tempurl}
	
	
	\bibitem[Yuan et~al\mbox{.}(2023)]%
	{YLW23}
	\bibfield{author}{\bibinfo{person}{Yuanyuan Yuan}, \bibinfo{person}{Zhibo Liu},
		{and} \bibinfo{person}{Shuai Wang}.} \bibinfo{year}{2023}\natexlab{}.
	\newblock \showarticletitle{CacheQL: Quantifying and Localizing Cache
		Side-Channel Vulnerabilities in Production Software}. In
	\bibinfo{booktitle}{\emph{Proc. of the 32nd {USENIX} Security Symposium}}.
	\newblock
	
	
	\bibitem[Zhang et~al\mbox{.}(2018)]%
	{ZGSW18}
	\bibfield{author}{\bibinfo{person}{Jun Zhang}, \bibinfo{person}{Pengfei Gao},
		\bibinfo{person}{Fu Song}, {and} \bibinfo{person}{Chao Wang}.}
	\bibinfo{year}{2018}\natexlab{}.
	\newblock \showarticletitle{SCInfer: Refinement-Based Verification of Software
		Countermeasures Against Side-Channel Attacks}. In
	\bibinfo{booktitle}{\emph{Proc. of the 30th International Conference on
			Computer Aided Verification}}. \bibinfo{pages}{157--177}.
	\newblock
	\urldef\tempurl%
	\url{https://doi.org/10.1007/978-3-319-96142-2\_12}
	\showDOI{\tempurl}
	
	
\end{thebibliography}
\end{document}